\documentclass[a4paper]{article}

\usepackage{pdfsync}
\usepackage[utf8]{inputenc}
\usepackage{graphicx}
\usepackage{amssymb,amsfonts,amsmath,amsthm}
\usepackage{hyperref}
\usepackage{multirow}
\usepackage{verbatim}
\usepackage[sectionbib,round]{natbib}
\usepackage{color}
\usepackage{bm}
\usepackage{rotating}
\usepackage{chngcntr}
\usepackage{orcidlink}

\theoremstyle{plain}
\newtheorem{theorem}{Theorem}[section]
\newtheorem{corollary}{Corollary}[section]

\newtheorem{proposition}{Proposition}[section]
\theoremstyle{definition}
\newtheorem{definition}{Definition}[section]

\newtheorem{remark}{Remark}[section]

\newcommand*{\defeq}{\mathrel{\vcenter{\baselineskip0.5ex \lineskiplimit0pt \hbox{\scriptsize.}\hbox{\scriptsize.}}}=}
\newcommand{\minus}{\scalebox{0.6}{$-$}}
\newcommand{\plus}{\scalebox{0.6}{$+$}}

\begin{document}

\title{\textbf{Data envelopment analysis with common-denominator ratio variables: An application to education with methodological extensions}}

\author{V. J. Bol\'os$^{1,*}$\,\orcidlink{0000-0002-9099-8200}, V. J. España$^2$\,\orcidlink{0000-0002-1807-6180}, R. Ben\'{\i}tez$^1$\,\orcidlink{0000-0002-9443-0209}, V. Coll-Serrano$^3$\,\orcidlink{0000-0002-7813-7379} \\
{\scriptsize $^1$ Dept. Business Mathematics, University of Valencia,
Avda. Naranjos s/n, Valencia, 46022, Spain.} \\
{\scriptsize $^2$ Dept. Statistics, Operations Research and Numerical Analysis, UNED, Paseo Senda del Rey 11, Madrid, 28040, Spain.} \\
{\scriptsize $^3$ Dept. Applied Economics, University of Valencia, Avda. Naranjos s/n, Valencia, 46022, Spain.} \\
{\scriptsize e-mail\textup{: \texttt{vicente.bolos@uv.es} (V. J. Bol\'os), \texttt{vjespana@ccia.uned.es} (V. J. España), \texttt{rabesua@uv.es} (R. Ben\'{\i}tez)},} \\
{\scriptsize \textup{\texttt{vicente.coll@uv.es} (V. Coll-Serrano)}} 
 \\}

\date{}

\maketitle

\begin{abstract}
The use of ratio variables in convex Data Envelopment Analysis (DEA) models has long been recognized as a methodological issue, as ratios may be incompatible with the convexity assumptions underlying standard DEA. However, in this paper, we prove that if all variables are ratio variables sharing a common denominator, then convex combinations of feasible activities are also feasible. Moreover, we establish a result demonstrating the equivalence between DEA models with common-denominator ratio variables under variable returns to scale and DEA models with volume (non-ratio) variables under constant returns to scale. These significant results enable the development of a framework for evaluating the efficiency of the Organisation for Economic Co-operation and Development (OECD) countries based on the results of the Programme for International Student Assessment (PISA) report, using mean performance scores as outputs. In this framework, we give some methodological extensions, such as the incorporation of the index of economic social and cultural status (ESCS) as an input, thereby enabling fairer comparisons with countries with a lower socio-economic level. Furthermore, we introduce different methods for estimating directions of improvement and calculating targets appropriate to the difficulty of improving each performance score. Finally, we review and introduce several novel contributions to emerging methodologies that can complement classical radial and directional models, such as efficient frontier estimation with adaptive constrained enveloping splines (ACES), stochastic chance-constrained models, and fuzzy models. All these methodologies can be used to analyse data from other PISA or similar reports, allowing non-specialists to implement DEA appropriately.
\end{abstract}


\section{Introduction}

Data envelopment analysis (DEA) is a non-parametric technique used to measure the relative efficiency of a set of homogeneous decision-making units (DMUs) that use multiple inputs to produce multiple outputs \citep{Charnes1978}. Using mathematical programming methods, DEA identifies the best practice frontier of the production possibility set, which is determined by efficient DMUs.
Radial models were introduced by \citet{Charnes1978, Charnes1979, Charnes1981} for constant returns to scale (CRS) (known as Charnes-Cooper-Rhodes (CCR) models), and \citet{Banker1984} for variable returns to scale (VRS) (known as Banker-Charnes-Cooper (BCC) models). In general, they can be input- or output-oriented. In the former case, the aim is to determine the maximum proportionate reduction in inputs permitted by the production possibility set while maintaining the current output level. Conversely, in the output-oriented case, the aim is to find the maximum proportionate increase in outputs while maintaining the current level of input consumption.

In the literature, radial models have been used to analyse the efficiency of the educational systems of the Organisation for Economic Co-operation and Development (OECD) countries, taking into account the average performance scores per student published by the Programme for International Student Assessment (PISA) as outputs. For example, \citet{Kocak2011,Cilin2018} used radial models with average mathematics, reading and science performance scores (from PISA 2009 report) as outputs, taking public education expenditure relative to gross domestic product (GDP) as the only input. \citet{Buy2022} also used radial models to analyse the same three PISA 2015 output scores, but with inputs from the 2013 Teaching and Learning International Survey (TALIS), specifically teaching time per week and several indices of teachers' self-efficacy and satisfaction.

Nevertheless, it is important to take into account that the use of ratio variables, such as averages or percentages, in conventional VRS and CRS DEA models can lead to serious problems in the interpretation of the results. The suitability of ratio data in these models has been a longstanding topic of academic discussion. Early contributions to this debate primarily emphasized recognizing the interpretation issues and attempted to determine which of the standard DEA models were most appropriate for handling such data \citep{Golany1997, Cooper2007}. Even so, the underlying technology is generally modeled incorrectly and it may include DMUs that cannot be produced. Several authors have investigated this issue, including \citet{Hollin03,Emr09,Olesen2015,Hatami19}. Among these contributions, \citet{Olesen2015} developed the R-VRS and R-CRS models of production technology suitable for ratio measures. Once the technologies are constructed, different projection methods can be used to obtain different DEA models and efficiency measures. However, although this technique is correct, it may exclude many activities from the production technology, and it can be too restrictive.

Aside from the debate over ratio variables, one of the limitations of radial models is that they require proportional reductions in inputs or proportional increases in outputs when computing the target activity. Directional models, introduced by \citet{Chambers1996,Chambers1998} and \citet{Briec1997}, overcome this drawback generalizing radial models by allowing the user to modify the proportion of inputs and outputs, thus enabling a custom orientation that takes into account the particularities of the market and the management criteria chosen by the producer. Other popular non-radial models are slacks-based measure (SBM) of efficiency models, introduced by \citet{Tone2001}. Nevertheless, these models are not as suitable as directional models for calculating the target for two reasons. Firstly, they do not allow the user to select the direction of improvement; secondly, the target obtained by SBM models is the furthest point on the efficient frontier that is dominated by the DMU being evaluated. This second drawback is addressed in SBM-Min models \citep{Tone2016}, but other problems, such as non-monotonicity, arise.

Despite its widespread use, traditional DEA is highly susceptible to overfitting, particularly when the dimensionality of the data is high or the sample size is small. Moreover, in many situations, DEA tends to generate overly optimistic efficiency scores, often evaluating units as efficient even when they are far from the true theoretical frontier. For example, under variable returns to scale, a DMU that attains the best value for a given input or output is always deemed efficient, regardless of the values of the remaining variables. To address these limitations, the integration of machine-learning (ML) techniques into frontier estimation has emerged as a robust alternative \citep{esteve2020, valero2021, guillen2023, moragues2023}.
However, these methods typically involve a substantially higher computational cost than standard linear programming approaches and require a careful hyperparameter selection process.
In this context, \citet{espana2024,espana2025} developed the adaptive constrained enveloping splines (ACES) methodology. ACES is a non-parametric technique based on an adaptation of the multivariate adaptive regression splines (MARS) algorithm introduced by \citet{friedman1991}.
Unlike conventional enveloping techniques, ACES employs piecewise linear spline functions and a backward pruning procedure guided by generalized cross-validation (GCV), which mitigates overfitting and significantly improves the model's ability to generalize beyond the observed sample.

Finally, in a stochastic framework, chance-constrained DEA \citep{Lan1993, Olesen1995} allows for uncertainty in inputs and outputs. This ensures that the resulting inefficiency scores are obtained with a given degree of confidence, providing a more robust and realistic assessment of efficiency than deterministic DEA. \citet{Cooper1996,Cooper1998,Cooper2002} developed chance-constrained versions of some radial models. More recently, \citet{Bolos2024} introduced chance-constrained directional models.
Another recent approach for imprecise data are fuzzy DEA models, with  Kao–Liu \citep{Kao2000}, Guo–Tanaka \citep{Guo2001} and possibilistic \citep{Leon2003} models as the most relevant. More recently, some chance-constrained fuzzy DEA models were introduced by \citet{Tavana2013}. Fuzzy DEA models have been shown to be very useful for analysing data where uncertainty arises from ambiguity in the variables. However, these methods do not make use of the information derived from the probability distribution of the data.

In this paper, we establish a methodology to apply radial and directional DEA models (classical deterministic, ACES, chance-constrained, fuzzy) to OECD countries as DMUs, using means of performance scores from the PISA 2022 report as outputs. We justify the suitability of this type of ratio variables proving an important and general result about the equivalence of VRS models with ratio variables and CRS models with volume (non-ratio) variables, provided that all variables are ratio variables sharing a common denominator. With respect to the methodological innovations, we establish a criterion to consider the index of \textit{economic social and cultural status} (ESCS) as an input in order to fairly compare the efficiency of educational systems of countries with students of different socio-economic statuses. In addition, we present various approaches for identifying improvement directions and determining targets that reflect the level of difficulty associated with enhancing each performance score. We also examine and propose several new contributions to emerging methodologies that can complement classical radial and directional models, including efficient frontier estimation using adaptive constrained enveloping splines (ACES), stochastic chance-constrained models, and fuzzy models.

The paper is organized as follows: Section \ref{sec:2} introduces the basic concepts of DEA, along with the theoretical and formal framework of the DEA models applied in this study. Classical radial and directional models are presented in Section \ref{sec:detmod}. Recent methodologies such as adaptive constrained enveloping splines (Section \ref{sec:aces_background}), and models with imprecise data (specifically, stochastic chance-constrained models and Kao–Liu fuzzy models in Section \ref{sec:imprecise}) are also considered. Next, we give some ideas to analyse efficient DMUs in Section \ref{sec:anaeff}. Section \ref{sec:ratio} presents the main results about common-denominator ratio variables in standard DEA models. The \textit{Theorem of equivalence for common-denominator ratio variables} is stated and proved in Section \ref{sec:teo}.
Section \ref{sec:2.1} presents an application to education using PISA data. Inputs and outputs are defined in Section \ref{sec:io}, the VRS assumption is discussed in Section \ref{sec:rts}, and two ways for achieving efficiency are presented in Section \ref{sec:models}. Next, Section \ref{sec:orientcoef} introduces some methods for estimating orientation coefficients in directional models, based on different criteria: an empirical approach (Section \ref{sec:orientcoef_empirical}) and a criterion based on potential improvement (Section \ref{sec:orientcoef_potential}). Section \ref{sec:results} presents the results of the illustrative example, including specific methodological procedures for applying ACES (Section \ref{sec:res_aces}) and the definition of fuzzy variables based on percentiles (Section \ref{sec:impreciseres}). Finally, Section \ref{sec:conclusions} offers concluding remarks. Data and results tables are provided in an appendix.

\section{DEA background}
\label{sec:2}

We consider $\mathcal{D}=\left\{ \textrm{DMU}_1, \ldots ,\textrm{DMU}_n \right\} $ a set of $n$ DMUs with $m$ inputs and $s$ outputs. Matrices $X=(x_{ij})$ and $Y=(y_{rj})$ are the \emph{input} and \emph{output data matrices}, respectively, where $x_{ij}>0$ and $y_{rj}>0$ refer to the $i$-th input and $r$-th output, respectively, of the $j$-th DMU.

The term \emph{activity} is used to describe any strictly positive vector of $m+s$ components. Any $\textrm{DMU}_o \in \mathcal{D}$ has its associated activity given by $(x_{1o},\ldots ,x_{mo};y_{1o}\ldots ,y_{so})$. Therefore, a DMU can be identified with its activity in the same way that a point is identified with its coordinates.
Given two activities $\mathbf{a}=(x_1,\ldots ,x_m;y_1\ldots ,y_s)$, $\mathbf{a}'=(x'_1,\ldots ,x'_m;y'_1\ldots ,y'_s)$, we say that $\mathbf{a}'$ \emph{dominates} $\mathbf{a}$ if $x'_i\leq x_i$ and $y_r\leq y'_r$ for $i=1,\ldots m$ and $r=1,\ldots ,s$. The relation ``to be dominated by'' is a partial order on the set of activities. Moreover, we say that $\mathbf{a}'$ \emph{strictly dominates} (or \emph{improves}) $\mathbf{a}$ if $x'_i<x_i$ and $y_r<y'_r$ for $i=1,\ldots m$ and $r=1,\ldots ,s$.

The \emph{production possibility set} (PPS) (also known as \textit{production technology}), denoted by $P$, is the set of all \emph{feasible activities} defined by $\mathcal{D}$. Assuming strong disposability of inputs and outputs, $P$ under variable returns to scale (VRS) is the set formed by activities dominated by convex combinations of DMUs in $\mathcal{D}$:
\begin{equation}
\label{eq:p}
\def\arraystretch{1.2}
P = \left\{(x_1,\ldots ,x_m;y_1\ldots ,y_s) \in \mathbb{R}_{>0}^{m+s}\ \ \Bigg| \
\begin{array}{ll}
\sum_{j=1}^n\lambda _jx_{ij}\leq x_i,& i=1,\ldots, m,  \\
\sum_{j=1}^n\lambda _jy_{rj}\geq y_r,& r=1,\ldots, s, \\
\sum _{j=1}^n \lambda _j=1,\, \lambda _j\geq 0,& j=1,\ldots ,n
\end{array} \right\} .
\end{equation}
Constant returns to scale (CRS) are assumed by removing the convexity condition in \eqref{eq:p}. 

We say that an activity $\mathbf{a}$ is \emph{efficient} (also known as \emph{strongly efficient} or \emph{Pareto-Koopmans efficient}) if there is not any other activity being feasible and dominating $\mathbf{a}$ \citep{Koop51}. Otherwise, we say that it is \emph{inefficient}.
The set of efficient activities in $P$ is called the \emph{efficient frontier} (also known as \emph{strongly efficient frontier} or \emph{Pareto-Koopmans frontier}) of $P$, and we denote it by $\partial ^{\text{S}}P$.
On the other hand, we say that an activity $\mathbf{a}$ is \emph{weakly efficient} (also known as \emph{technically efficient} or \emph{Farrell efficient}) if there is not any other activity being feasible and strictly dominating $\mathbf{a}$ \citep{Deb51,Far57}. The set of weakly efficient activities in $P$ is called the \emph{weakly efficient frontier} of $P$, and we denote it by $\partial ^{\text{W}}P$. It is clear that $\partial ^{\text{S}}P\subset \partial ^{\text{W}}P$, but they are not the same.
We assume strong disposability of inputs and outputs and, hence, the weakly efficient frontier of $P$ coincides with the boundary of $P$, i.e. $\partial ^{\text{W}}P=\partial P$ \citep{Fare1985}.
Note that this fact makes weakly efficient activities interesting (even if they are inefficient) because they compound the boundary of the PPS, and they can be interpreted as the ``limits of the technology''.

\subsection{Classical radial and directional models}
\label{sec:detmod}

In this section, we present the classical radial and directional models with deterministic variables, in their output-oriented versions under VRS. The corresponding CRS version is constructed eliminating the convexity constraint $\sum _{j=1}^n \lambda _j=1$. All models evaluate efficiency of a given $\textrm{DMU}_o\in \mathcal{D}$. 

Radial models were originally defined by \citet{Charnes1978}. The output-oriented version under VRS is given by
\begin{equation}
\label{eq:rad}
\def\arraystretch{1.2}
\begin{array}[t]{rll}
\eta ^*=\max \limits_{\eta ,\lambda _1,\ldots ,\lambda _n} & \eta \\
\textrm{s.t.} & \sum_{j=1}^n\lambda _jx_{ij}\leq x_{io},& \quad i=1,\ldots ,m, \\
& \sum_{j=1}^n\lambda _jy_{rj}-\eta y_{ro}\geq 0,& \quad r=1,\ldots ,s, \\
& \sum _{j=1}^n \lambda _j=1,\quad \lambda _j\geq 0,& \quad j=1,\ldots ,n.
\end{array}
\end{equation}
The \emph{target} in $\partial ^{\text{W}}P$ is defined as
$\left( x_{1o},\ldots ,x_{mo};\eta ^*y_{1o},\ldots ,\eta ^*y_{so}\right) $, and the \textit{inefficiency slacks} are given by the optimal values of the slack variables associated with the inequality constraints in \eqref{eq:rad}.
Once $\eta ^*$ has been determined, a \textit{second-stage} optimization problem is solved to obtain the so-called \textit{max-slack solution}, which maximizes the sum of inefficiency slacks while keeping $\eta ^*$ fixed. The \emph{efficient projection} in $\partial ^{\text{S}}P$ is given by
$\sum_{j=1}^n\lambda ^*_j\left( x_{1j},\ldots ,x_{mj};y_{1j},\ldots, y_{sj}\right) $, where $\lambda ^*_1,\ldots ,\lambda ^*_n $ are optimal for the second-stage program. The \emph{reference set} for $\textrm{DMU}_o$ consists of those $\textrm{DMU}_j$ for which $\lambda _j^*>0$.

If $\eta ^*=1$, then $\textrm{DMU}_o$ coincides with its target and hence, $\textrm{DMU}_o$ is weakly efficient. If, in addition, the efficient projection coincides with $\textrm{DMU}_o$, then $\textrm{DMU}_o$ is efficient. In this case, the target also coincides with $\textrm{DMU}_o$ and all the inefficiency slacks are $0$. On the other hand, if $\eta ^*>1$, then $\textrm{DMU}_o$ is inefficient. In this case, the inefficiency score is $\eta ^*$.

Directional models generalize radial models and were introduced by \citet{Chambers1996,Chambers1998}. Output-oriented directional models are written in terms of the output \textit{direction coefficients} $g_1,\ldots ,g_s\geq 0$ corresponding to the evaluated DMU. However, we are going to consider the formulation given by \citet{Briec1997} in terms of the so called output \textit{orientation coefficients} $d_1,\ldots ,d_s\geq 0$. The relation between the direction and the orientation coefficients is $g_r=d_r y_{ro}$ for $r=1,\ldots ,s$. The output-oriented directional model under VRS is given by
\begin{equation}
\label{eq:dird}
\def\arraystretch{1.2}
\begin{array}[t]{rll}
\beta ^*=\max \limits_{\beta ,\lambda _1,\ldots ,\lambda _n} & \beta \\
\textrm{s.t.} & \sum_{j=1}^n\lambda _jx_{ij}\leq x_{io},& \quad i=1,\ldots ,m, \\
& \sum_{j=1}^n\lambda _jy_{rj}-\beta d_ry_{ro}\geq y_{ro},& \quad r=1,\ldots ,s, \\
& \sum _{j=1}^n \lambda _j=1,\quad \lambda _j\geq 0,& \quad j=1,\ldots ,n.
\end{array}
\end{equation}
The \emph{target} in $\partial ^{\text{W}}P$ is given by
$\left( x_{1o},\ldots ,x_{mo};(1+\beta ^*d_1)y_{1o},\ldots ,(1+\beta ^*d_s)y_{so}\right) $, and the \textit{inefficiency slacks} are defined in the same way as in the radial case. A \textit{second-stage} optimization problem is also solved to identify the \textit{max-slack solution}. The \emph{efficient projection} and the \emph{reference set} are also defined analogously.
If $\beta ^*=0$, then $\textrm{DMU}_o$ coincides with its target and hence, $\textrm{DMU}_o$ is weakly efficient. If, in addition, the efficient projection coincides with $\textrm{DMU}_o$, then $\textrm{DMU}_o$ is efficient. On the other hand, if $\beta ^*>0$, then $\textrm{DMU}_o$ is inefficient. In this case, the \emph{inefficiency score} is $\beta ^*$. As a homogeneous efficiency score, the corresponding \emph{Farrell oriented efficiency} \citep{Bolos2025} is given by
$\rho ^*=1/\left(1+\beta ^*\frac{1}{s}\sum _{r=1}^s d_r\right) $.
We have $0<\rho ^*\leq 1$, and $\rho ^*=1$ if and only if $\textrm{DMU}_o$ is its own target.

The output orientation coefficients $d_1,\ldots ,d_s$ in program \eqref{eq:dird} determine the proportions by which the outputs of $\textrm{DMU}_o$ are dilated to calculate the target. In the particular case $d_1=\cdots =d_s=d>0$, all the outputs are dilated in the same proportion and hence, we say that the directional model is radial. In fact, we obtain the classical output-oriented radial model taking $\eta =1+\beta d$ and, in this case, the \emph{Farrell radial efficiency} is given by $\rho ^*=1/\eta ^*$. Moreover, if the orientation coefficient of an output is equal to $0$, then that output can not be changed in the calculation of the target.
In any case, the output orientation coefficients characterize the criteria of management chosen by the producer for the evaluated DMU, representing the relative ``ease of improvement'' of each output, and therefore determining an improvement strategy: the higher the orientation coefficient of an output, the easier it is to improve that output.

\subsection{Adaptive constrained enveloping splines}
\label{sec:aces_background}

In this section, we briefly present the ACES methodology introduced by \citet{espana2024} and further developed by \citet{espana2025}, in its output-oriented version under VRS. ACES is a non-parametric frontier estimation method based on an adaptation of the multivariate adaptive regression splines algorithm of \citet{friedman1991}. In multi-input and multi-output settings, the method follows three main steps. First, it estimates one regularized upper frontier for each output, ensuring that the observed sample is enveloped from above. Second, it applies a conservative refinement step to avoid artificial overestimations in the multi-output setting. Finally, a DEA-type technology under VRS is constructed by replacing the original outputs with the refined ACES predictions. Standard efficiency measures, including radial and directional models, can then be computed with respect to this regularized technology.

ACES first expands the original input vector by allowing interaction terms among inputs. Let $\mathfrak{g}(\boldsymbol{x})\in\mathbb{R}^{m^\ast}_{>0}$ denote the augmented vector obtained from the original $m$ inputs and the selected multiplicative interactions, where $m^\ast$ is the total number of generated features. The frontier is then represented through hinge basis functions (BFs). For each component $x_i$ of the augmented vector and each knot $k$, ACES considers the functions $B^+(x_i;k)=(x_i-k)_+$, $B^-(x_i;k)=(k-x_i)_+$, where the subscript $( \cdot )_+$ denotes the positive part operator, i.e., $\max\{0, \cdot \}$. The candidate knots are taken from the observed values of each variable, so that the model can adapt locally to the empirical structure of the data.

For each output $r=1,\ldots,s$, ACES estimates a frontier
$\hat f^{\text{ACES}}_r:\mathbb{R}_{>0}^{m^\ast}\to\mathbb{R}_{>0}$. All output-specific frontiers share the same selected BF and knot structure, although their coefficients may differ. Hence, for each DMU $j$,
\begin{equation*}
\hat f^{\text{ACES}}_r(\mathfrak{g}(\boldsymbol{x}_j)) =
\tau^{(r)} +
\sum_{i=1}^{m^{\ast}}\sum_{h\in H_i}\alpha^{(r)}_{i_h}
\bigl(x_{ij}-\kappa^{(R)}_{i_h}\bigr)_+ +
\sum_{i=1}^{m^{\ast}}\sum_{w\in W_i}\beta^{(r)}_{i_w}
\bigl(\kappa^{(L)}_{i_w}-x_{ij}\bigr)_+,
\end{equation*}
where $\tau^{(r)}$ is the intercept, $\alpha^{(r)}_{i_h}$ and $\beta^{(r)}_{i_w}$ are slope coefficients, and $H_i$ and $W_i$ denote the sets of active right- and left-sided BFs for feature $i$.

The algorithm combines a forward construction stage and a backward pruning stage. Starting from a constant model, the forward stage iteratively adds the reflected pair of BFs that produces the largest improvement according to a lack-of-fit criterion, such as the mean squared error (MSE) or mean absolute error (MAE). This process stops when the maximum number of BFs, $T_{\max}$, is reached or when the relative improvement falls below a threshold $\xi$. Since the forward stage may generate an overfitted model, the backward stage removes BFs sequentially and selects the final specification through the generalized cross-validation criterion
\begin{equation}
\label{eq:gcv}
\text{GCV} =
\left.
\left(
\frac{1}{n} \sum_{j=1}^{n} \sum_{r=1}^{s}
\left(y_{rj} - \hat{f}_r(\mathfrak{g}(\boldsymbol{x}_j))\right)^2
\right)
\middle/
\left(1-\frac{C(\mathcal{B}) + \gamma \cdot \chi(\mathcal{B})}{n}\right)^2
\right.
\end{equation}
where $C(\mathcal{B})$ is the number of fitted coefficients, $\chi(\mathcal{B})$ is the number of knots, and $\gamma$ controls the penalty for model complexity.

At each step of the forward and backward stages, the coefficients are estimated by solving a linear programming problem. The model minimizes a weighted sum of non-negative residuals while imposing envelopment, monotonicity and concavity. Since the estimator is piecewise linear, its partial derivative with respect to each input is constant within each interval defined by the selected knots (for details, see \cite{espana2025}. The optimization problem is
\begin{equation} \label{eq:aces_multioutput_lp}
\begin{aligned}[t]
\min_{\substack{\boldsymbol{\varepsilon}, \{\tau^{(r)}, \alpha^{(r)}, \beta^{(r)}\}_{r=1, \dots, s.}}} \quad & \sum_{j=1}^{n} \omega_j \sum_{r=1}^{s} \varepsilon_{rj} \\
\text{s.t.} 
\quad & \hat f^{\text{ACES}}_r \bigl(\mathfrak{g}(\boldsymbol{x}_j)\bigr) - \varepsilon_{rj} = y_{rj}, && \forall j, \forall r, \\
& \varepsilon_{rj} \ge 0, && \forall j, \forall r, \\
& \left.\partial_{i}\hat f^{\text{ACES}}_r\right|_{v} \ge 0, && \forall r, \forall i, v=1,...,|K_i|+1, \\
& \left.\partial_{i}\hat f^{\text{ACES}}_r\right|_{v} - \left.\partial_{i}\hat f^{\text{ACES}}_r\right|_{v+1} \ge 0, && \forall r, \forall i, v=1,...,|K_i|.
\end{aligned}
\end{equation}

The objective function employs inverse-efficiency weights ($\omega_j = 1/\eta_j$) to prioritize observations near the technological boundary, preventing distortions from highly inefficient interior units. The constraints guarantee envelopment ($\varepsilon_{rj} \ge 0$), forcing the estimated frontier to act as a strict theoretical ceiling consistent with technical inefficiency. Furthermore, the core microeconomic axioms are imposed dimensionally across the knots: non-decreasing monotonicity via $\left.\partial_{i}\hat f^{\text{ACES}}_r\right|_{v} \ge 0$, and concavity (non-increasing marginal returns) via $\left.\partial_{i}\hat f^{\text{ACES}}_r\right|_{v} - \left.\partial_{i}\hat f^{\text{ACES}}_r\right|_{v+1} \ge 0$.

A specific difficulty in the multi-output version of ACES is that output frontiers are estimated separately. Consequently, the vector of predicted outputs may not always be consistent with the joint production structure. To avoid artificial overestimations, \citet{espana2025} propose the following refinement rule:
\begin{equation} \label{eq:refinement}
    \hat{y}_{rj} = \begin{cases} 
    \hat{f}_r^{\text{ACES}}(\mathfrak{g}(\boldsymbol{x}_j)) & \text{if }\quad | \eta ^*_j - \eta ^{*(r)}_j | < \psi \\
    y_{rj} & \text{if }\quad | \eta ^*_j - \eta ^{*(r)}_j | \geq \psi 
    \end{cases},
    \qquad r=1,\ldots ,s,\quad j=1,\ldots ,n,
\end{equation}
where $\eta ^*_j$ is the radial output-oriented efficiency score obtained with all outputs, $\eta ^{*(r)}_j$ is the corresponding score obtained using only output $r$, and $\psi$ is a tolerance threshold.

Once the refined outputs $(\hat{y}_{1j}, \dots, \hat{y}_{sj})$ are obtained for each $\textrm{DMU}_j$, artificial DMUs are constructed by combining the original inputs with the refined outputs. The resulting set $\hat{\mathcal{D}}$ defines a new VRS PPS that can be used to evaluate any $\textrm{DMU}_o\in\mathcal{D}$ by considering $\hat{\mathcal{D}}\cup{\textrm{DMU}_o}$ as the reference set. Since this technology is convex and satisfies free disposability by construction, standard technical efficiency measures, including radial, directional and slacks-based models, can be computed with respect to a regularized frontier consistent with economic theory.

The ACES algorithm relies on hyperparameters to balance frontier flexibility and theoretical consistency. In the forward stage, $T_{\max}$ and $\xi$ limit the maximum number of BF and the minimum relative improvement for inclusion, respectively. Admissible knots are governed by \textit{minspan} and \textit{endspan}, while $q_{\max}$ and $\xi^{(q)}$ control interaction terms. During pruning, $\gamma$ scales the GCV complexity penalty. While coefficients are estimated via the constrained LP \eqref{eq:aces_multioutput_lp} to minimize non-negative residuals, forward selection ranks candidate BF using a specified lack-of-fit (\textit{LOF}) metric, such as MAE or MSE. Finally, the multi-output refinement threshold $\psi$ in \eqref{eq:refinement} conservatively dictates whether an ACES prediction is retained. Because the resulting virtual DEA-type technology intrinsically satisfies standard axioms, imposing monotonicity and concavity during the initial fit serves as a structural tool to enhance frontier estimation rather than a strict theoretical requirement.

Finally, it should be noted that hyperparameter selection requires validation tailored to frontier estimation. Because frontier models recover an upper boundary where residuals reflect latent inefficiency rather than noise, standard cross-validation often overfits inefficient interior observations and distorts the frontier. Thus, ACES validation must employ frontier-oriented criteria, such as weighted errors or trimming strategies focused on observations near the empirical boundary.

\subsection{Models with imprecise data}
\label{sec:imprecise}

In this section, uncertainty will be handled with two different approaches: stochastic and fuzzy models. In the stochastic case, we will use chance-constrained models, while on the latter, we will consider Kao--Liu fuzzy meta-models. 

In chance-constrained models, the constraints hold with a probability greater than or equal to $1-\alpha$ (with $0<\alpha \leq 0.5$), assuming that inputs and outputs are stochastic variables. Hence, the chance-constrained version of \eqref{eq:rad} is given by
\begin{equation}
\label{eq:radcc}
\def\arraystretch{1.5}
\begin{array}[t]{rll}
\eta ^*_{\textrm{CC}}=\max \limits_{\eta ,\lambda _1,\ldots ,\lambda _n} & \eta \\
\textrm{s.t.} & P\left\{ \sum_{j=1}^n\lambda _j\tilde{x}_{ij}\leq \tilde{x}_{io}\right\} \geq 1-\alpha ,& \quad i=1,\ldots ,m, \\
& P\left\{ \sum_{j=1}^n\lambda _j\tilde{y}_{rj}-\eta \tilde{y}_{ro}\geq 0\right\} \geq 1-\alpha ,& \quad r=1,\ldots ,s, \\
& \sum _{j=1}^n \lambda _j=1,\quad \lambda _j\geq 0,& \quad j=1,\ldots ,n,
\end{array}
\end{equation}
where $\tilde{x}_{ij}$ and $\tilde{y}_{rj}$ represent the $i$-th input and the $r$-th output of the $j$-th DMU. Usually, inputs and outputs are considered to be correlated random variables following multivariate normal distributions. In this way, the deterministic equivalent of \eqref{eq:radcc} is introduced by \citep{Cooper1996}.

Analogously, the chance-constrained version of \eqref{eq:dird} is given by
\begin{equation}
\label{eq:dirdcc}
\def\arraystretch{1.5}
\begin{array}[t]{rll}
\beta ^*_{\textrm{CC}}=\max \limits_{\beta ,\lambda _1,\ldots ,\lambda _n} & \beta \\
\textrm{s.t.} & P\left\{ \sum_{j=1}^n\lambda _j\tilde{x}_{ij}\leq \tilde{x}_{io}\right\} \geq 1-\alpha ,& \quad i=1,\ldots ,m, \\
& P\left\{ \sum_{j=1}^n\lambda _j\tilde{y}_{rj}-\beta d_r\tilde{y}_{ro}\geq \tilde{y}_{ro}\right\} \geq 1-\alpha ,& \quad r=1,\ldots ,s, \\
& \sum _{j=1}^n \lambda _j=1,\quad \lambda _j\geq 0,& \quad j=1,\ldots ,n.
\end{array}
\end{equation} 
The deterministic equivalent of \eqref{eq:dirdcc} is introduced by \citep{Bolos2024}, for which \textit{stochastic directions} are considered. Note that we will consider the improvement strategy for outputs to be given in terms relative to the stochastic performance scores, and therefore the direction of improvement, given by $\left( d_1\tilde{y}_{1o},\ldots ,d_s\tilde{y}_{so}\right) $, is also stochastic. Moreover, it is worth noting that the deterministic equivalents of \eqref{eq:radcc} and \eqref{eq:dirdcc} are quadratic optimization problems, which entail a higher computational cost compared to linear models.

Fuzzy models can handle imprecise data in the form of fuzzy numbers.
A \emph{fuzzy set} $A$ is defined by a function $\mu _A:\mathbb{R}\rightarrow \left[ 0,1\right] $, called \emph{membership function}. 
For any $x\in \mathbb{R}$, the value $\mu _A(x)$ can be interpreted as the grade of membership of $x$ to $A$. Alternatively, a fuzzy set is completely determined by the so-called $\alpha$-\emph{cuts} which are defined by $A^{\alpha}=\left\{ x\in \mathbb{R}\,\,|\,\,\mu_A(x) \geq \alpha \right\} $ for $\alpha \in \left] 0,1\right]$, and $A^0=\overline{\left\{ x\in \mathbb{R}\,\,|\,\,\mu_A(x)>0\right\}}$, where the overline denotes clausure.
\emph{Fuzzy numbers} are a particular case of fuzzy sets verifying $A^{\alpha}$ convex (i.e., intervals) for $\alpha \in \left[ 0,1\right] $ and $A^1\neq \emptyset$ (normalized).

The most relevant fuzzy DEA models are Kao–Liu \citep{Kao2000}, Guo–Tanaka \citep{Guo2001} and possibilistic \citep{Leon2003} models. However, Guo-Tanaka models can only be applied to triangular symmetric data under CRS, and possibilistic models can only be applied to trapezoidal data, which is not our case as we are going to see in Section \ref{sec:impreciseres}.
On the other hand, Kao--Liu models consist on applying an existing model in the worst and best scenarios for the evaluated DMU at each $\alpha$-cut. Therefore, they can be considered meta-models. The worst scenario for a DMU occurs when it uses the largest possible input amounts and obtains the smallest possible output amounts, while, on the contrary, the rest of the DMUs use the smallest possible input amounts, obtaining the largest possible output amounts. For example, considering the radial model \eqref{eq:rad}, the worst scenario for DMU$_o$ at a given $\alpha$-cut is determined by
\begin{equation}
\label{eq:radworst}
\def\arraystretch{1.2}
\begin{array}[t]{rll}
\eta ^*_{\textrm{W}}=\max \limits_{\eta ,\lambda _1,\ldots ,\lambda _n} & \eta \\
\textrm{s.t.} & \sum_{j=1,j\neq o}^n\lambda _jx^{\textrm{L}}_{ij}+\lambda _o x^{\textrm{U}}_{io}\leq x^{\textrm{U}}_{io},& \quad i=1,\ldots ,m, \\
& \sum_{j=1,j\neq o}^n\lambda _jy^{\textrm{U}}_{rj}+\lambda _o y^{\textrm{L}}_{ro}-\eta y^{\textrm{L}}_{ro}\geq 0,& \quad r=1,\ldots ,s, \\
& \sum _{j=1}^n \lambda _j=1,\quad \lambda _j\geq 0,& \quad j=1,\ldots ,n,
\end{array}
\end{equation}
where superscripts L and U denote the lower and upper extremes, respectively, of the corresponding $\alpha$-cut for each fuzzy variable.
On the other hand, the best scenario for a DMU occurs when it uses the smallest possible input amounts and obtains the largest possible output amounts, while the rest of the DMUs use the largest possible input amounts, obtaining the smallest possible output amounts. Considering the radial model \eqref{eq:rad}, the best scenario for DMU$_o$ at a given $\alpha$-cut is determined by
\begin{equation}
\label{eq:radbest}
\def\arraystretch{1.2}
\begin{array}[t]{rll}
\eta ^*_{\textrm{B}}=\max \limits_{\eta ,\lambda _1,\ldots ,\lambda _n} & \eta \\
\textrm{s.t.} & \sum_{j=1,j\neq o}^n\lambda _jx^{\textrm{U}}_{ij}+\lambda _o x^{\textrm{L}}_{io}\leq x^{\textrm{L}}_{io},& \quad i=1,\ldots ,m, \\
& \sum_{j=1,j\neq o}^n\lambda _jy^{\textrm{L}}_{rj}+\lambda _o y^{\textrm{U}}_{ro}-\eta y^{\textrm{U}}_{ro}\geq 0,& \quad r=1,\ldots ,s, \\
& \sum _{j=1}^n \lambda _j=1,\quad \lambda _j\geq 0,& \quad j=1,\ldots ,n.
\end{array}
\end{equation}
In this paper, radial and directional models given by \eqref{eq:rad} and \eqref{eq:dird}, respectively, are applied in Kao--Liu meta-models. 

\subsection{Analysis of efficient DMUs}
\label{sec:anaeff}

Let $\textrm{DMU}_j\in \mathcal{D}$ be efficient, a way to estimate its ``influence'' on the rest of DMUs is to count in how many reference sets of inefficient DMUs it appears, i.e. for how many inefficient DMUs there is a strictly positive optimal value $\lambda ^*_j$. A more accurate way to estimate this influence is to calculate the mean of all these optimal values $\lambda ^*_j$ when we vary the DMU to be evaluated. It can be expressed as a percentage, the sum of which is $100\%$ under VRS (see Table~\ref{tab:lambdas}, last row). This value, interpreted as \textit{influence}, represents the relative importance of each efficient DMU in the reference sets of inefficient DMUs, and hence in the construction of efficient projections of other DMUs.

Another feature to consider in efficient DMUs is how much surplus they have, and it can be studied applying \emph{super-efficiency} models \citep{Andersen1993}. 
Nevertheless, the calculation of super-efficiency scores of some efficient DMUs may reveal unfeasibilities under VRS, specifically those DMUs that have some input or output better than the rest of the DMUs.

\section{Common-denominator ratio variables in DEA}
\label{sec:ratio}

A potential concern arises from the use of ratios as inputs and/or outputs in conventional DEA models. In this case, convex combinations of DMUs may not remain feasible, thereby undermining one of the fundamental assumptions of DEA. This issue is illustrated in \cite[Example 2]{Olesen2015}. The problem arises because, when forming a convex combination of DMUs, the resulting ratio variables generally do not coincide with the correct ratios associated with that combination of DMUs. Hence, there can be activities in the PPS \eqref{eq:p} that are not feasible in practice or are incorrect. This is especially common when ratio and volume variables are considered together, or the nature of ratio variables is different, such as percentages, averages, proportions, etc.

In these cases, it is necessary to define a more restrictive PPS based on activities that are known to be feasible.
This work has been done by \citet{Olesen2015} for a wide variety of ratio variable types, developing the R-VRS and R-CRS models of production technology that allow the use of ratio measures ``as is'', without any transformation. However, although this technique is valid, it can be overly restrictive, as it may exclude from the PPS activities that are feasible in practice.

Nevertheless, Proposition \ref{prop:ratio} assures that that convex combinations of DMUs remain feasible when all variables are expressed as ratios with a common denominator for each DMU, without requiring any modification or adaptation.

\begin{proposition}
\label{prop:ratio}
Let us suppose that all variables (both inputs and outputs) are common-denomin\-ator ratio variables (i.e. they are expressed as ratios sharing a common denominator, which may differ across DMUs). Then, the ratios obtained from a convex combination of DMUs are consistent with the correct ratios corresponding to another convex combination of those same DMUs.
\end{proposition}

\begin{proof}
Let us explain and prove this property in terms of the application to education in Section \ref{sec:2.1}:

Let $(p_1,\ldots ,p_n)\in \mathbb{R}^n_{>0}$ be the number of students in $n$ DMUs (countries), and let $(a_1,\ldots ,a_n)\in \mathbb{R}^n_{>0}$ represent average per-student values of some magnitude. Given $\lambda _1,\ldots ,\lambda _n\geq 0$ such that $\sum _{j=1}^n \lambda _j=1$, the convex combination of the averages, $\sum _{j=1}^n \lambda _j a_j$, do not correspond to the real average of the convex combination $\sum _{j=1}^n \lambda _j \textrm{DMU}_j$, which is given by
$\sum _{j=1}^n \lambda _j p_j a_j / \sum _{j=1}^n \lambda _j p_j$.
For example, if DMU$_1$ has $p_1=200$ students with a mean performance score of $a_1=9$, and DMU$_2$ has $p_2=1800$ students with a mean performance score of $a_2=5$, then $0.5\cdot \textrm{DMU}_1+0.5\cdot \textrm{DMU}_2$ has $0.5\cdot p_1+0.5\cdot p_2=1000$ students with a real mean performance score of $(0.5\cdot p_1 a_1+0.5\cdot p_2 a_2)/(0.5\cdot p_1+0.5\cdot p_2)=5.4$, 
and not $0.5\cdot a_1+0.5\cdot a_2=7$.

However, there exist $\tilde{\lambda}_1,\ldots ,\tilde{\lambda}_n\geq 0$ such that
\begin{equation}
\label{eq:lambdaa}
    \sum _{j=1}^n \lambda _j a_j=\sum _{j=1}^n \tilde{\lambda}_j p_j a_j \bigg/ \sum _{j=1}^n \tilde{\lambda}_j p_j,\qquad \sum _{j=1}^n \tilde{\lambda}_j=1.
\end{equation}
The left-hand side in \eqref{eq:lambdaa} is the original convex combination of the averages, while the right-hand side is the correct average of the convex combination $\sum _{j=1}^n \tilde{\lambda}_j \textrm{DMU}_j$. In particular, it is easy to prove that
\begin{equation}
\label{eq:lambdatilde}
    \tilde{\lambda}_j=\lambda _j\bigg/ p_j \sum _{k=1}^n \frac{\lambda _k}{p_k},\quad j=1,\ldots ,n.
\end{equation}
Returning to our application, although $7$ is not the correct mean performance score of $0.5\cdot \textrm{DMU}_1+0.5\cdot \textrm{DMU}_2$, it is the correct mean performance score of $0.9\cdot \textrm{DMU}_1+0.1\cdot \textrm{DMU}_2$.

According to \eqref{eq:lambdatilde}, note that $\tilde{\lambda}_j$ does not depend on the average values of the magnitude, $a_1,\ldots ,a_n$, and thus, \eqref{eq:lambdaa} also holds for averages of several magnitudes at once. That is, considering $m$ magnitudes, let $a_{rj}>0$ be the average value of the $r$-th magnitude for DMU$_j$, with $r=1\ldots ,m$ and $j=1,\ldots ,n$. Then, $\sum _{j=1}^n\lambda _j\left( a_{1j},\ldots ,a_{mj}\right) $ are the real average values of all $m$ magnitudes for the convex combination $\sum _{j=1}^n \tilde{\lambda}_j \textrm{DMU}_j$. Given that, in our case, all variables (both inputs and outputs) are ratios sharing a common denominator, it follows that the convex combination of ratios $\sum _{j=1}^n\lambda _j\left( x_{1j},\ldots ,x_{mj};y_{1j},\ldots ,y_{sj}\right) $ is correct for the convex combination of DMUs $\sum _{j=1}^n \tilde{\lambda}_j \textrm{DMU}_j$.
\end{proof}

Proposition \ref{prop:ratio} establishes that convex combinations of feasible activities are also feasible when all variables are common-denominator ratio variables. This result provides the basis for constructing the PPS under VRS in the following corollary.

\begin{corollary}
\label{cor:pps}
Let us suppose that all variables are common-denominator ratio variables and that input variables are not upper-bounded. Then, assuming strong disposability of inputs and outputs, the corresponding PPS under VRS is given by \eqref{eq:p}.
\end{corollary}
\begin{proof}
By Proposition \ref{prop:ratio}, the PPS \eqref{eq:p} under VRS, which is the set of activities dominated by convex combinations of ratios, is indeed equivalent to the set of activities dominated by correct ratios of convex combinations of DMUs. Since the input variables are not upper-bounded and both inputs and outputs are strongly disposable, the PPS \eqref{eq:p} under VRS is the set of all feasible activities.
\end{proof}

\begin{remark}
If the $i$-th input is upper-bounded by $b_i$, the constraint $x_i \leq b_i$ must be added to the constraints defining the PPS.
\end{remark}

The homogeneity of all variables, in the sense that they share a common denominator (for example, ratios expressed on a per-student basis), is crucial for the correctness of convex combinations.
The resulting PPS \eqref{eq:p} under VRS, constructed from activities dominated by correct ratios of convex combinations of DMUs in $\mathcal{D}$, accurately represents the set of all feasible activities, assuming that inputs are not upper bounded and strong disposability of inputs and outputs. In this case, the models developed by \citet{Olesen2015} would exclude from the PPS activities that are actually feasible. In fact, taking into account that all variables are ratio variables, the PPS from the R-VRS model is given by
\begin{equation*}
\def\arraystretch{1.2}
\left\{(x_1,\ldots ,x_m;y_1\ldots ,y_s) \in \mathbb{R}_{>0}^{m+s}\ \ \Bigg| \
\begin{array}{ll}
\lambda _j\left( x_{ij}-x_i\right) \leq 0,& i=1,\ldots, m,\quad j=1,\ldots ,n,  \\
\lambda _j\left( y_{rj}-y_r\right) \geq 0,& r=1,\ldots, s,\quad j=1,\ldots ,n, \\
\sum _{j=1}^n \lambda _j=1,\, \lambda _j\geq 0,& j=1,\ldots ,n
\end{array} \right\} ,
\end{equation*}
which is, obviously, more restrictive than \eqref{eq:p} (see Figure~\ref{fig:pps}). In conclusion, Proposition \ref{prop:ratio} and Corollary \ref{cor:pps} allow the application of classical radial and directional models under VRS without requiring modifications to accommodate ratio variables.

\begin{figure}[htbp]
\centering
\includegraphics[width=0.6\textwidth]{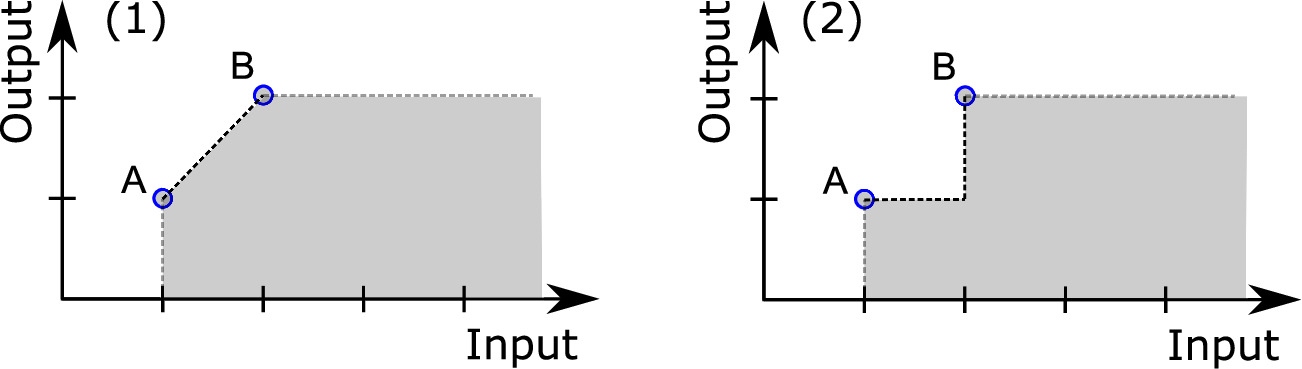}
\caption{Representation of the PPS (grey area) generated by two DMUs ($A$ and $B$) with one input (assumed to be unbounded above) and one output, both common-denominator ratio variables and assumed to be strongly disposable: (1) considering a classical VRS model, and (2) considering the R-VRS model of \citet{Olesen2015}. The corresponding efficient frontier is represented by black dashed lines, while the rest of the weakly efficient frontier is represented by grey dashed lines. From Corollary \ref{cor:pps}, PPS (1) coincides with the set of all feasible activities, whereas PPS (2) excludes some of them.}
\label{fig:pps}
\end{figure}

\subsection{Theorem of equivalence for common-denominator ratio variables}
\label{sec:teo}

In this section, we continue assuming that all variables (inputs and outputs) are ratio variables sharing a common denominator for each DMU, namely $p_1,\ldots, p_n$, which represents the number of students of each DMU in our application. Moreover, we assume that, when the corresponding volume variables $\tilde{x}_{ij}=x_{ij}p_j$, $\tilde{y}_{rj}=y_{rj}p_j$ are considered, $p_j$ can be regarded as a non-controllable variable for DMU$_j$.

\emph{Non-controllable} variables constitute a special category of variables that are exogenously determined and cannot be modified by the decision maker. They may be treated as inputs, outputs, or neither. \emph{Non-discretionary} variables are also exogenously fixed, but they differ from non-controllable variables in an important respect. According to \citet{Cooper2007}, two activities can only be compared when they have the same values of the non-controllable variables. In contrast, activities with different values of non-discretionary variables, whether inputs or outputs, can still be compared, thereby allowing dominance relationships to be established between them. Consequently, non-controllable variables affect the definition of the PPS itself, whereas non-discretionary variables affect only the resulting efficiency scores.

Several approaches have been proposed in the literature to account for exogenously fixed variables \citep{BM86,Lovell94,Ruggiero96}. Here, however, we adopt the approach proposed by \citet{Cooper2007}, according to which, assuming strong disposability of inputs and outputs, the PPS under CRS with inputs $\tilde{x}_1,\ldots,\tilde{x}_m$, outputs $\tilde{y}_1,\ldots,\tilde{y}_s$, and a non-controllable variable $p$, is given by
\begin{equation}
\label{eq:pnc}
\def\arraystretch{1.2}
\left\{(\tilde{x}_1,\ldots ,\tilde{x}_m;\tilde{y}_1\ldots ,\tilde{y}_s;p) \in \mathbb{R}_{>0}^{m+s+1}\ \ \Bigg| \
\begin{array}{ll}
\sum_{j=1}^n\tilde{\lambda} _j\tilde{x}_{ij}\leq \tilde{x}_i,& i=1,\ldots, m,  \\
\sum_{j=1}^n\tilde{\lambda} _j\tilde{y}_{rj}\geq \tilde{y}_r,& r=1,\ldots, s, \\
\sum _{j=1}^n \tilde{\lambda} _jp_j=p,\, \tilde{\lambda} _j\geq 0,& j=1,\ldots ,n
\end{array} \right\} .
\end{equation}
Note that the PPS \eqref{eq:pnc} is formed by activities dominated by positive combinations of DMUs, but the non-controllable variable has to be equal to the corresponding positive combination of the values of the non-controllable variables of the DMUs in order to be comparable.

\begin{proposition}
\label{prop:prethm}
Let us consider a set of $n$ DMUs with $m$ inputs, $s$ outputs, and a non-controllable variable. Then, the PPS under CRS given by \eqref{eq:pnc} is transformed into the PPS under VRS given by \eqref{eq:p} when ratio variables are considered, dividing the original inputs and outputs by the non-controllable variable.
\end{proposition}

\begin{proof}
Let us consider
\begin{equation}
\label{eq:tildelambdaprop}
    \lambda _j=\frac{p_j}{p_o}\tilde{\lambda} _j,\qquad x_{ij}=\frac{\tilde{x}_{ij}}{p_j},\qquad y_{rj}=\dfrac{\tilde{y}_{rj}}{p_j}.
\end{equation}
Then, applying \eqref{eq:tildelambdaprop} to \eqref{eq:pnc}, we obtain \eqref{eq:p}.
\end{proof}

Proposition \ref{prop:prethm} justifies the use of the VRS assumption for common-denominator ratio variables when the corresponding volume variables are appropriately modeled under CRS and the denominator is a non-controllable variable. An important implication of this result is that a CRS technology with volume inputs and outputs and a non-controllable variable can be expressed in terms of a VRS technology by normalizing all inputs and outputs with respect to the non-controllable variable. This makes it possible to apply methods that require the VRS assumption, such as archetypal analysis of DMUs \citep{alcacer2025survey}, which cannot be applied under the CRS assumption.

Next, we go one step further. Theorem \ref{thm} shows that program \eqref{eq:dird}, and consequently program \eqref{eq:rad}, is equivalent to the program resulting from the corresponding CRS model with volume variables $\tilde{x}_{ij}=x_{ij}p_j$, $\tilde{y}_{rj}=y_{rj}p_j$, and $p_j$ considered as a non-controllable variable for DMU$_j$. Furthermore, this equivalence holds for any ``well-constructed'' DEA model.

\begin{definition}
    We say that a DEA model has an \textit{envelopment} formulation if it computes an efficient projection in $\partial ^{\text{S}}P$ of the form $\sum_{j=1}^n\lambda ^*_j\left( x_{1j},\ldots ,x_{mj};y_{1j},\ldots, y_{sj}\right) $, with $\lambda ^*_1,\ldots ,\lambda ^*_n $ being optimal for the corresponding program model.  
\end{definition}

\begin{definition}
    Let $\Omega ^*_o\left( \mathcal{D}\right) $ denote the optimal objective function of a DEA model under CRS for the evaluation of DMU$_o$ with respect to $\mathcal{D}=\left\{ \textrm{DMU}_1, \ldots ,\textrm{DMU}_n \right\} $. We say that this model is \textit{CRS-invariant} if $\Omega ^*_o\left( \left\{ \alpha _1\textrm{DMU}_1, \ldots ,\alpha _n\textrm{DMU}_n \right\} \right) =\Omega ^*_o\left( \mathcal{D}\right) $ for any $(\alpha _1,\ldots ,\alpha _n)\in \mathbb{R}^n_{>0}$.
\end{definition}

Any ``well-constructed'' unit-invariant model under CRS should be CRS-invariant. For example, assuming CRS, classical radial and directional models are CRS-invariant, but so are SBM models \citep{Tone2001} and non-radial models from \citet{Fare1978}. Moreover, additive models \citep{Charnes1985} are CRS-invariant only if the slack weights in the objective function are chosen such that the model is unit-invariant.

\begin{theorem}[\textbf{Theorem of equivalence for common-denominator ratio variables}]
\label{thm}
Let us consider a set of $n$ DMUs with $m$ inputs, $s$ outputs, and a non-controllable variable. Then, a CRS-invariant model applied to those $m+s+1$ variables is equivalent to the same model, but under VRS and applied to $m+s$ ratio variables, constructed by dividing the original inputs and outputs by the non-controllable variable. 
\end{theorem}

\begin{proof}
Let $\tilde{x}_{ij}$, $\tilde{y}_{rj}$ and $p_j$ be the $i$-th input, the $r$-th output and the non-controllable variable, respectively, of the $j$-th DMU. We are going to prove the result for a general directional model under CRS evaluating DMU$_o$. The corresponding program, taking into account the PPS \eqref{eq:pnc}, is given by
\begin{equation}
\label{eq:dirdcrs}
\def\arraystretch{1.2}
\begin{array}[t]{rll}
\beta ^*=\max \limits_{\beta ,\tilde{\lambda} _1,\ldots ,\tilde{\lambda} _n} & \beta \\
\textrm{s.t.} & \sum_{j=1}^n\tilde{\lambda} _j\tilde{x}_{ij}+\beta d_i^{\minus}\tilde{x}_{io}\leq \tilde{x}_{io},& \quad i=1,\ldots ,m, \\
& \sum_{j=1}^n\tilde{\lambda} _j\tilde{y}_{rj}-\beta d_r^{\plus}\tilde{y}_{ro}\geq \tilde{y}_{ro},& \quad r=1,\ldots ,s, \\
& \sum _{j=1}^n \tilde{\lambda} _jp_j=p_o,\quad \tilde{\lambda} _j\geq 0,& \quad j=1,\ldots ,n,
\end{array}
\end{equation}
where $d_i^{\minus}$ and $d_r^{\plus}$ are the input and output orientation coefficients, respectively. Note that the constraints $\sum _{j=1}^n \tilde{\lambda} _jp_j=p_o$ are added due to the non-controllable nature of $p_j$.
Let us consider the same change of variables as in Proposition \ref{prop:prethm}:
\begin{equation}
\label{eq:tildelambda}
    \lambda _j=\frac{p_j}{p_o}\tilde{\lambda} _j,\qquad x_{ij}=\frac{\tilde{x}_{ij}}{p_j},\qquad y_{rj}=\dfrac{\tilde{y}_{rj}}{p_j}.
\end{equation}
Then, applying \eqref{eq:tildelambda} to program \eqref{eq:dirdcrs}, we obtain
\begin{equation}
\label{eq:tildedirdcrs}
\def\arraystretch{1.2}
\begin{array}[t]{rll}
\beta ^*=\max \limits_{\beta ,\lambda _1,\ldots ,\lambda _n} & \beta \\
\textrm{s.t.} & \sum_{j=1}^n\lambda _jx_{ij}+\beta d_i^{\minus}x_{io}\leq x_{io},& \quad i=1,\ldots ,m, \\
& \sum_{j=1}^n\lambda _jy_{rj}-\beta d_r^{\plus}y_{ro}\geq y_{ro},& \quad r=1,\ldots ,s, \\
& \sum _{j=1}^n \lambda _j=1,\quad \lambda _j\geq 0,& \quad j=1,\ldots ,n.
\end{array}
\end{equation}
Program \eqref{eq:tildedirdcrs} corresponds to the same directional model as that of program \eqref{eq:dirdcrs}, but under VRS and applied to the ratio inputs $x_{ij}$ and ratio outputs $y_{rj}$, without incorporating the non-controllable variable $p_j$. Since radial models constitute a particular case of directional models, the proof is also valid for radial models.

Furthermore, the same change of variables \eqref{eq:tildelambda} can be used for any other CRS-invariant models (such as SBM models), since the ratio variables in \eqref{eq:tildelambda} are proportional to the original volume variables for each DMU. Moreover, the program constraint for variable $p_j$ to be non-controllable, namely $\sum _{j=1}^n\tilde{\lambda}_jp_j=p_o$, is the same for all models in its envelopment formulation. This constraint is transformed into the convexity VRS condition $\sum_{j=1}^n\lambda_j=1$ applying \eqref{eq:tildelambda}.
\end{proof}

\begin{remark}
The non-controllable variable in Theorem \ref{thm} can also serve as an input or output under the CRS assumption. In this case, when ratio variables are considered, the corresponding non-controllable input or output takes an all-ones vector.
\end{remark}

Next, in Section \ref{sec:2.1}, we present an application to education using PISA data, with common-denominator ratio variables expressed on a per-student basis. In this setting, if the corresponding volume variables $\tilde{x}_{ij}=x_{ij}p_j$ and $\tilde{y}_{rj}=y_{rj}p_j$ are considered, both the CRS assumption and the treatment of the number of students $p_j$ as a non-controllable variable are fully justified, as explained in Section \ref{sec:rts}.

Although Proposition \ref{prop:prethm} and Theorem \ref{thm} can be applied by considering \textit{number of students} as the non-controllable variable, these results are applicable to a wide range of other settings. For example, when companies are considered as DMUs, \textit{number of employees} can be treated as a non-controllable variable, alongside the inputs \textit{total assets}, \textit{staff costs}, and \textit{other costs}, and the output \textit{operating revenues} \citep{Bas2026}. In sports analytics, when basketball players are considered as DMUs, \textit{minutes played} can likewise be treated as a non-controllable input, while the totals of \textit{offensive rebounds}, \textit{defensive rebounds}, \textit{assists}, \textit{steals}, and \textit{points} are considered as outputs. Other examples of non-controllable variables in different contexts include \textit{number of patients} in hospitals, \textit{number of customers} in retail businesses, \textit{number of researchers} in university departments, \textit{population} in cities, etc.

\section{An application to PISA data}
\label{sec:2.1}

Next, we present an illustrative example based on the 2022 PISA report. In this example, the DMUs consist of countries that are members of the OECD. Costa Rica, Estonia, New Zealand, Switzerland, and the United Kingdom are excluded from the analysis due to incomplete data availability. Consequently, the sample consists of the following 32 countries:
Australia, Austria, Belgium, Canada, Chile, Colombia, the Czech Republic, Denmark, Finland, France, Germany, Greece, Hungary, Iceland, Ireland, Israel, Italy, Japan, Korea, Latvia, Lithuania, Mexico, the Netherlands, Norway, Poland, Portugal, the Slovak Republic, Slovenia, Spain, Sweden, T\"urkiye, and the United States. 

\subsection{Inputs and outputs}
\label{sec:io}

The outputs considered in this study are the \textit{mean scores in mathematics} \citep[Table I.B1.2.1]{pisa2022}, \textit{reading} \citep[Table I.B1.2.2]{pisa2022}, and \textit{science} \citep[Table I.B1.2.3]{pisa2022}, as published in the 2022 PISA report. These outputs may also be treated as stochastic variables, each with its standard deviation. In the existing literature, such measures constitute the most commonly used outputs for evaluating the efficiency of education systems based on PISA data.

Input selection is not trivial and depends on the kind of analysis that is intended.
In the literature, a wide variety of inputs has been considered, including variables that are not reported in the PISA databases. For example, \citet{Buy2022} considers instruction time and several indices of teachers' self-efficacy and satisfaction published in the 2013 Teaching and Learning International Survey (TALIS). \citet{espana2025rf} also consider inputs for human and capital resources, specifically the teacher-student ratio and the quality of school resources, along with the index of economic social and cultural status (ESCS).

The aim of this paper is not to determine which inputs are most appropriate, but rather to provide a framework that allows researchers to apply radial and directional models with the inputs they consider relevant. Nevertheless, some technical aspects concerning the choice of inputs must be taken into account. In particular, when average performance scores are considered as outputs, only variables expressed on a per-student basis or as student-related ratios are admissible.
By contrast, the inclusion of variables of a different nature, such as volume variables, percentages, or indices of teachers’ satisfaction, may generate inconsistencies in the definition of the PPS. In such cases, the guidelines proposed by \citet{Olesen2015} should be followed. 

Based on the ideas of \citet{Kocak2011} and \citet{Buy2022}, our illustrative example evaluates the efficiency of the educational system by incorporating two inputs that represent human and capital resources:
\begin{itemize}

\item \textit{Instruction time}: compulsory instruction time in general education per student (primary + lower secondary) \citep[Table D1.1]{glance2023}. Mexico data is taken from \citep[Table D1.1]{glance2025}.

\item \textit{Expenditure}: estimated cumulative expenditure on educational institutions per student (from age 6 to 15) \citep[Table B3.2.2]{pisa2022}.

\end{itemize}

Note that the expenditure data reported in PISA are already adjusted using purchasing power parities (PPPs) for GDP, thereby eliminating the bias associated with differences in purchasing power across countries.

In addition to these variables (namely, instruction time per student and expenditure per student) other variables such as the GDP per student, the teacher-to-student ratio or the computer-to-student ratio could also be considered.

Moreover, it is known that the students' economic, social and cultural background significantly affects performance scores \citep[Chapter 4]{pisa2022}\citep{Sirin2005}.  
Hence, in order to mitigate this bias, we are going to make a parallel analysis considering the mean value of the \textit{index of economic social and cultural status (ESCS)} \citep[Table I.B1.4.2]{pisa2022} as another input, jointly with \textit{instruction time} and \textit{expenditure}.
The same argument could be applied to select GDP per student as an input, because countries with high GDP are supposed to have better infrastructures and it also affects performance scores. However, the ESCS is more appropriate because it takes social factors into account, rather than focusing solely on economic or resource-related aspects.

The ESCS of a student is a composite score that combines into a single score information from three components: parents' highest level of education, parents' highest occupational status, and home possessions (which is a proxy for family wealth). The higher the value of ESCS, the higher the socio-economic status of the student. The ESCS index of each country is estimated as the mean value of the ESCS for all the students in the country.
Since the ESCS is a contextual factor which is external to schools, it is a good input for measuring the intrinsic efficiency of the educational system of each country. Moreover, it is proved that the ESCS affects school performances \citep[Chapter 4]{pisa2022}\citep{Sirin2005} and considering it as an input would eliminate this bias.

Nevertheless, the ESCS index is constructed to have a mean of $0$ across OECD countries, implying that some countries exhibit negative ESCS values. This feature is not suitable for use as an input in conventional DEA models. To address this issue, the ESCS index must be transformed into a strictly positive variable by applying a translation of magnitude $\tau > 0$. Note that, for the output-oriented radial and directional models under VRS, given by \eqref{eq:rad} and \eqref{eq:dird}, the value of $\tau$ is irrelevant, provided that $\tau$ is chosen such that all translated ESCS values are positive \citep{Ali90}. This follows because, under the VRS condition $\sum_{j=1}^n \lambda_j=1$, the constraints $\sum_{j=1}^n\lambda_j(x_{ij}+\tau)\leq x_{io}+\tau$ reduce to $\sum_{j=1}^n\lambda_jx_{ij}\leq x_{io}$. Moreover, ESCS can be regarded as a stochastic variable, as standard deviations are provided for this index. Note that these standard deviations are unaffected by the translation. Since standard deviations of ESCS are in the range of approximately 0.7 to 1.2 \citep[Table I.B1.4.2]{pisa2022}, we set $\tau = 2-\min_{j=1,\ldots,n}{\mathrm{ESCS}_j}$, where $\mathrm{ESCS}_j$ denotes the ESCS index for $\mathrm{DMU}_j$. Consequently, the translated ESCS index is greater than or equal to 2, while its standard deviation remains in the range of approximately 0.7 to 1.2, making it a suitable input for the DEA models.

\begin{remark}[Non-discretionary ESCS]
\label{rem:escsnd}
The ESCS can be considered as a non-discretionary input in the sense of \cite{Cooper2007}, since it is exogenously fixed and therefore, it is not possible to vary it at the discretion of management. Nevertheless, since we are going to consider output-oriented models, this consideration has no effect \citep{Cooper2007}.
\end{remark}

Finally, we note that, for computational reasons, it is preferable for all variables to have a similar order of magnitude. Therefore, instruction time is measured in hundreds of hours and expenditure on education is measured in thousands of USD (converted using PPPs for GDP). Lastly, the (translated) ESCS is also multiplied by $100$ (see Table~\ref{tab:data}).

In conclusion, the use of radial and directional VRS models is fully justified in this application because all inputs and outputs are ratio variables sharing a common denominator, namely the number of students. Moreover, the number of students need not be known, in contrast to the equivalent CRS formulation based on volume variables.

\subsection{Returns to scale}
\label{sec:rts}

We consider common-denominator ratio variables expressed on a per-student basis. In this section, we justify the use of the VRS assumption in this setting by showing that the conditions required to apply Proposition \ref{prop:prethm} and Theorem \ref{thm} are satisfied.
Specifically, when the corresponding volume variables $\tilde{x}_{ij}=x_{ij}p_j$, $\tilde{y}_{rj}=y_{rj}p_j$ are considered, both the CRS assumption and the treatment of the number of students $p_j$ as a non-controllable variable are fully justified.

First, when volume variables are considered, it is reasonable to assume that if inputs and the number of students are multiplied by a factor, outputs will be also multiplied by the same factor. This is equivalent to assuming that, for example, with twice as many students, the number of schools could also be doubled. Provided that the same teaching methodology is maintained across all schools, students would achieve the same grades. Consequently, the total sum (volume) of all grades would also be twice as large.

Second, the number of students should be treated as a non‑controllable variable because it is exogenously determined, and it is only meaningful to compare activities involving the same number of students. For example, suppose DMU A achieves a grade volume (output) of 100 points (the sum of all grades) with an instruction time volume (input) of 100 hours and 100 students (i.e. $1$ hour per student and $1$ point per student), while DMU B achieves the same grade volume with the same instruction time volume but with 200 students (i.e. $0.5$ hours per student and $0.5$ points per student). In order to show that the number of students is a non-controllable variable when volume variables are considered, we propose two wrong scenarios:
\begin{itemize}
    \item The number of students is a controllable input: then, A dominates B. So, A is considered better, because it obtains the same amount of outputs with fewer inputs.
    \item The number of students is a controllable output: then, B dominates A. So, B is considered better.
\end{itemize}
However, both conclusions are incorrect because, in practice, we cannot directly determine which DMU is better. To make a valid comparison, we must consider proportional (scaled) activities with the same number of students. For instance, if we scale up A to 200 students, it would achieve a grade volume of 200 points with an instruction time volume of 200 hours. This scaled activity can now be meaningfully compared with B, since both involve 200 students. We then observe that the scaled A achieves a higher grade volume (200 vs. 100), whereas B uses a lower volume of instruction time (100 vs. 200). Therefore, neither activity dominates the other.

\subsection{The road to efficiency}
\label{sec:models}

Our focus is on improving mean performance scores (outputs), rather than reducing inputs, in order to reach the boundary of the PPS. Therefore, we consider output-oriented models.

Moreover, all variables are ratio variables sharing a common denominator, namely the number of students. The corresponding volume variables are appropriately modeled under CRS, and the number of students is a non-controllable variable in this setting. Therefore, as established in Section \ref{sec:rts}, the VRS assumption is appropriate for our application.

Any DEA model in this work calculates a \emph{target} activity in $\partial ^{\text{W}}P$ (following the improvement strategy) and \emph{efficient projections} in $\partial ^{\text{S}}P$.
Moreover, an \emph{efficiency score} $0<\rho ^*\leq 1$ can be computed which equals $1$ for activities in $\partial ^{\text{W}}P$. Taking into account that the outputs are mean performance scores, there exist two alternative ways for a country to attain efficiency when its target in $\partial^{\text{W}}P$ is not efficient (see Figure~\ref{fig:tw}):

\begin{figure}[htbp]
\centering
  \includegraphics[width=0.6\textwidth]{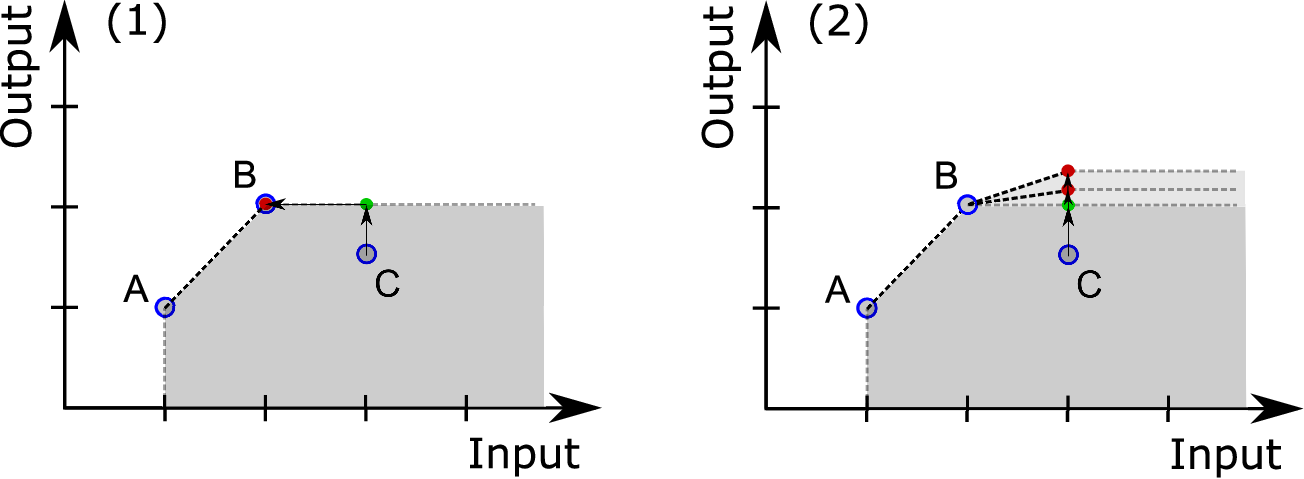}
\caption{Representation of two ways in which a DMU can achieve efficiency if its target is not efficient. We consider an output oriented model with one input and one output under VRS. The PPS $P$ is the grey area. The efficient DMUs are $A$ and $B$, which define the efficient frontier $\partial ^{\text{S}}P$ represented by a black dashed line.  The rest of the weakly efficient frontier $\partial ^{\text{W}}P$ is represented by grey dashed lines. Since the target of $C$ (green circle) is not efficient, $C$ can achieve efficiency by either (1) matching the efficient projection (red circle), or (2) matching the target and improving the output (red circles) in a marginal quantity. In the latter case, the PPS increases in line with the extra improvement in the output (note that this can only be achieved if the nature of the output allows for improvement even if it is in the frontier of the PPS, as in our case); however, there is no need for a reduction in the input.}
\label{fig:tw}
\end{figure}

\begin{itemize}
\item[(1)] Improve its activity so as to coincide with the efficient projection in $\partial^{\text{S}}P$. In this case, the PPS remains unchanged, but additional slacks in one or more variables (including inputs) may be required. This adjustment may be conceptually inappropriate in certain contexts; for instance, if it entails a reduction in ESCS.

\item[(2)] Improve its activity to reach the target in $\partial^{\text{W}}P$, and subsequently implement an additional improvement in at least one output variable with zero inefficiency slack. Although this additional improvement may be arbitrarily small, larger increases in the output will result in higher levels of super-efficiency. To preserve the underlying improvement strategy, other outputs may also be proportionally increased according to the output orientation coefficients (see Section~\ref{sec:detmod}). Under this approach, the PPS expands, while outputs improve consistently with the improvement strategy.
\end{itemize}

The expansion of the PPS relative to the original technology may be undesirable in certain settings, particularly in production processes where the level of inputs imposes a theoretical upper bound on achievable outputs. However, this limitation does not apply in the present context, as it is theoretically possible to marginally increase mean performance scores (provided that the maximum score has not been reached) while holding input levels constant.

To illustrate, consider the radial model under VRS, with instruction time, expenditure, and ESCS as inputs, while performance scores are the outputs. For Spain, the target activity has the original inputs $(79.23, 229.34, 315.22)$ (see Table~\ref{tab:data}) and improved outputs $(508.01, 509.26, 520.24)$ (see Table~\ref{tab:res}). This target is weakly efficient, but inefficient. By contrast, the activity of the efficient projection corresponds to the convex combination $0.308\cdot \text{Ireland} + 0.581\cdot \text{Japan} + 0.111\cdot \text{T\"urkiye}$ (see Table~\ref{tab:lambdas}), yielding input levels $(74.74, 93.09, 315.22)$ and output levels $(512.89, 509.26, 525.60)$. Hence, Spain may achieve efficiency through the following alternatives:
\begin{itemize}
\item[(1)] Improve its performance scores to $(512.89, 509.26, 525.60)$ while reducing instruction time to $74.74$. In this case, output improvements are non-radial, and efficiency requires an additional contraction in instruction time.

\item[(2)] Improve its performance scores to $(508.01, 509.26, 520.24)$ and subsequently introduce a mar\-gin\-al improvement in reading performance, which is the only output with zero inefficiency slack. To preserve the radial improvement structure (where all outputs increase proportionally) this additional improvement may be extended to mathematics and science performances according to the same proportions. For instance, an increase of $1$ point in reading performance may be accompanied by increases of $0.9975$ and $1.0216$ points in mathematics and science performances, respectively. The resulting efficient activity then attains output levels $(509.00, 510.26, 521.26)$. Under this approach, the PPS expands in a manner consistent with the present context (i.e., marginal improvements in mean performance scores), while maintaining radial output adjustments and avoiding any reduction in inputs.
\end{itemize}

A more detailed examination of the second alternative can be conducted in the non-radial framework, where the output orientation coefficients are defined by \eqref{eq:diremp}. In this case, Spain’s mean performance scores can be adjusted to the target $(503.44, 509.26, 514.13)$ (see Table~\ref{tab:targetsvrs134}), followed by marginal improvements proportional to the orientation coefficients $(2.208, 2.541, 2.106)$ (see Table~\ref{tab:dir}). For example, additional increases of $0.86$, $1$, and $0.84$ points in mathematics, reading, and science performances, respectively, yield an efficient activity with outputs $(503.30, 510.26, 514,97)$. This activity adheres to the prescribed improvement strategy defined by the orientation coefficients, without requiring any reduction in input levels.

\subsection{Output orientation coefficients}
\label{sec:orientcoef}

The output orientation coefficients $d_1,\ldots ,d_s$ can be interpreted as the relative ``ease of improvement'' of each output of the evaluated DMU, and therefore they establish an improvement strategy for $\textrm{DMU}_o$, determining the proportions by which outputs are dilated in the calculation of the target.

Deciding the value of these coefficients is not trivial and, although they do not affect efficient or weakly efficient DMUs, they significantly influence the calculation of targets for inefficient DMUs. We are going to present several methods for determining the output orientation coefficients by means of quantifying the relative ``ease of improvement'' of each performance score for each DMU based on data in PISA. Nevertheless, it should be noted that the aforementioned methods are merely illustrative of the estimation of these parameters. Depending on the researcher's criteria or the information available, there may be a multitude of ways to estimate them. 

\subsubsection{Empirical: annual rate of change}
\label{sec:orientcoef_empirical}

Let $\dot{y}_{rj}$ be the estimated rate of change of the $r$-th otuput of $\textrm{DMU}_j$, that we consider annual taking years as time units. This information is provided by the PISA 2022 report \citep{pisa2022} for mathematics (Table I.B1.5.4.), reading (Table I.B1.5.5.) and science (Table I.B1.5.6.) performances, under the name ``Annual rate of change in 2022 (linear term)''. According to this report, they are computed only for countries with comparable data in more than four PISA assessments, and they are calculated by a regression of performance over linear terms of the gap between the year of assessment and 2022. In our case, there is no problem because this rate is computed for all countries of OECD.

So, we must decide the values of $d_1,\ldots ,d_s$, interpreted as the relative ease of improvement of each output of $\textrm{DMU}_o$, based on the rates of change $\dot{y}_{rj}$ for $r=1,\ldots ,s$ and $j\in eval$, where $eval\subseteq \left\{ 1,\ldots ,n\right\} $ contains the indices of the DMUs that we want to evaluate and compare with $\textrm{DMU}_o$. Since $\dot{y}_{rj}$ are determined from previous experience, we can say that this approach is empirical.

First, we normalise the rates of change dividing by the values of the corresponding outputs, i.e. considering $\dot{y}_{rj}/y_{rj}$ instead of $\dot{y}_{rj}$. Table~\ref{tab:dir} shows these relative annual rates of change in 2022 of each performance.

Next, we set an \emph{orientation threshold}, $\delta \geq 0$, and define
$m_{eval}\defeq \min _{\substack{r=1,\ldots ,s\\ j\in eval}}\left\{ \dot{y}_{rj}/y_{rj}\right\}$.
If $m_{eval}\geq \delta $, then we define
$d_r\defeq \frac{\dot{y}_{ro}}{y_{ro}}$, for $r=1,\ldots s$.
On the other hand, if $m_{eval}<\delta $, then we define
\begin{equation}
\label{eq:diremp}
d_r\defeq \frac{\dot{y}_{ro}}{y_{ro}}+\delta -m_{eval},\quad r=1,\ldots s.
\end{equation}
In this way, it is assured that $d_r\geq \delta $ for all $r=1,\ldots ,s$.

The orientation threshold $\delta $ allows us to compare relative rates of change when there are non positive (or, more precisely, lower than $\delta $) values, translating them to become greater than or equal to $\delta $. If $m_{eval}\geq \delta $, then $\delta $ is irrelevant. But if $m_{eval}<\delta $, then we translate all the relative rates of change by $\delta -m_{eval}$, bringing the lowest to $\delta $ and thus, converting non positive rates into positive. If $\delta =0$ then the orientation coefficient of the output with the lowest relative rate of change is $0$ and hence, this output does not change in the calculation of the target. As $\delta $ increases, the orientation coefficients become more relatively homogeneous and the model tends to be radial (output-oriented).

It is recommended to take the same orientation threshold for all the DMUs that we want to evaluate. In this way, the researcher must choose an appropriate $\delta $ that  reflects the differences in the ease of improvement of each output of each $\textrm{DMU}_j$ with $j\in eval$, according to the relative annual rates of change of the outputs. In most cases, a balanced choice of $\delta $ is around half of the maximum value of the direction coefficients in the case $\delta =0$, but it may depend on the nature of the variables. For example, Table~\ref{tab:dir} shows the output orientation coefficients of all DMUs, computed with orientation threshold $\delta =1\%$. In this case, the maximum value of a direction coefficient with $\delta =0$ is $2.152\%$ (T\"urkiye $d_3$) and $\delta =1\%$ is a reasonable choice. 

\subsubsection{Potential: high and low percentile differences}
\label{sec:orientcoef_potential}

Since outputs are stochastic, a more theoretical method for estimating the relative ease of improvement of each output can be based on high and low percentiles. The underlying idea is that the greater the distance between a high percentile (90th or 75th percentile, for example) and the median, the greater the theoretical margin for improvement of the output. Moreover, we can also take into account the corresponding symmetric low percentile (10th or 25th percentile), which will act as a counterweight to this theoretical improvement.
All these data, i.e. 10th, 25th, 50th (median), 75th and 90th percentiles, are provided by the PISA 2022 report \citep{pisa2022} for mathematics (Table I.B1.2.1.), reading (Table I.B1.2.2.) and science (Table I.B1.2.3.) performances (see Table~\ref{tab:fuzzyper}).

In this section, we will use as high and low percentiles the $90$th and $10$th percentiles, respectively. However, it can be also considered the $75$th and $25$th percentiles, giving more attenuated differences. Let us denote by $\textbf{P}^{rj}_x$ the $x$th percentile of the $r$th output of $\textrm{DMU}_j$. We define
$h_{rj}\defeq \frac{\textbf{P}^{rj}_{90}-\textbf{P}^{rj}_{50}}{\textbf{P}^{rj}_{50}}$, and $l_{rj}\defeq \frac{\textbf{P}^{rj}_{50}-\textbf{P}^{rj}_{10}}{\textbf{P}^{rj}_{50}}$,
for $r=1,\ldots ,s$ and $j=1,\ldots ,n$.
In a first approach, we can take $h_{rj}$ as an indicator of the relative capacity of improvement, without taking into account $l_{rj}$. Since this indicator is always positive, we define for $\textrm{DMU}_o$
\begin{equation}
\label{eq:p90}
d_r\defeq h_{ro},\qquad r=1,\ldots ,s.
\end{equation}
Table~\ref{tab:dir} shows these orientation coefficients for all DMUs.

However, if we want to take into account low percentiles, we must consider $h_{rj}-l_{rj}$ as indicator instead of $h_{rj}$. We have to take into account that this new indicator can be negative and hence, we must apply a methodology analogous to that of the previous section. That is, we set an orientation threshold $\delta \geq 0$ and define
$m_{eval}\defeq \min _{\substack{r=1,\ldots ,s\\ j\in eval}}\left\{ h_{rj}-l_{rj} \right\} $,
where $eval\subseteq \left\{ 1,\ldots ,n\right\} $ contains the indices of the DMUs that we want to evaluate and compare with $\textrm{DMU}_o$. If $m_{eval}\geq \delta $, we define
$d_r\defeq h_{ro}-l_{ro}$, for $r=1,\ldots ,s$.
On the other hand, if $m_{eval}<\delta $, we define
\begin{equation}
\label{eq:p9010}
d_r\defeq h_{ro}-l_{ro}+\delta -m_{eval},\qquad r=1,\ldots ,s.
\end{equation}
Table~\ref{tab:dir} shows the output orientation coefficients of all DMUs, computed with orientation threshold $\delta =5\%$. According to the criterion used in the previous section, this value of $\delta $ is a reasonable choice. It should also be noted that mathematics is the performance area in which it is easiest for almost all countries to improve.

\subsection{Analysis and results of the illustrative example}
\label{sec:results}

In this section, we will apply radial and directional DEA models with deterministic and imprecise variables (chance-constrained and Kao--Liu fuzzy models), together with the ACES framework, demonstrating the key findings. Furthermore, the analysis is conducted both with ESCS included as an input and without its consideration, allowing for a comparative assessment of the results. It is important to note that weakly efficient countries are classified as efficient across all models applied to our dataset of 32 OECD countries. However, these properties do not necessarily hold in general. All computations are made using R 4.5.1 \citep{R25} with package deaR 1.5.4 \citep{deaR25, deaRSX} for radial, directional and Kao--Liu fuzzy models, and package SdeaR 1.0.2 \citep{SdeaR25} for chance-constrained models. The ACES-based results were obtained from our own implementation in R using functions available in our public repository at \url{https://github.com/Victor-Espana/aces}. All data and scripts are publicly available at \url{https://www.uv.es/vbolos/investigacion/script_PISA.zip}, enabling full replication of the results.

\subsubsection{Classical radial and directional models}
\label{sec:detvar}

We observe no significant differences between the efficiency scores $\rho^*$ obtained from radial models \eqref{eq:rad} and those derived from directional models \eqref{eq:dird} (see Table~\ref{tab:res}). When the ESCS is not included as an input, the efficient countries (defined by $\rho^* = 1$ and targets coinciding with their efficient projections, see Tables~\ref{tab:targetsvrs13}, \ref{tab:targetsvrs134}, \ref{tab:lambdas}) are Ireland, Japan, Korea, Latvia, Mexico, Poland, and T\"urkiye. 
This result can be partly explained by the input and output profiles of these countries. Poland has the lowest level of instruction time, Mexico has the lowest level of expenditure, and T\"urkiye has relatively low levels of both instruction time and expenditure. Latvia also combines relatively low input levels with comparatively high performance scores. Ireland exhibits moderately low input levels together with high performance scores. Japan attains the highest performance scores, while Korea ranks second, with slightly lower performance scores and less instruction time than Japan.
When the ESCS is included as an input, Colombia also becomes efficient. This may be related to its relatively low ESCS value, which is the second lowest among the countries considered, after T\"urkiye.

Regarding the contribution of each efficient country to the reference sets of the other countries, Japan and Ireland have the largest shares. Specifically, Japan accounts for $43.95\%$ (without ESCS) and $53.15\%$ (with ESCS), while Ireland contributes $24.54\%$ (without ESCS) and $18.17\%$ (with ESCS) (see Table~\ref{tab:lambdas}). Notably, Ireland’s influence decreases when ESCS is incorporated as an input, whereas  Japan's contribution increases under this specification.. 
Korea and T\"urkiye also display moderate and relatively stable levels of contribution, accounting for $11.47\%$ (without ESCS) and $8.92\%$ (with ESCS) for Korea, and $10.52\%$ (without ESCS) and $10.92\%$ (with ESCS) for T\"urkiye. Korea’s relatively limited influence may be attributed to the dominance of Japan, which demonstrates very similar, albeit slightly higher, performance scores. In contrast, Poland, Latvia, and Mexico exhibit the smallest contributions. Poland accounts for $4.29\%$ (without ESCS) and $5.59\%$ (with ESCS), Latvia for $3.12\%$ (without ESCS) and $3.25\%$ (with ESCS), and Mexico for $2.11\%$ (without ESCS) and $0\%$ (with ESCS). Finally, as in the case of Mexico, Colombia, which is efficient only when ESCS is included as an input, has a contribution of $0\%$. The zero contribution of Mexico and Colombia can be related to their combination of the lowest expenditure levels and the lowest performance scores across all output dimensions, which places them at extreme positions in the PPS. Consequently, although they are efficient, their positions are not used to construct the efficient projections of the other countries.

The efficient frontier is determined by the above efficient countries, and the inefficiency of the remaining countries is assessed relative to this benchmark. Countries located close to the PPS frontier (with $\rho^* \geq 0.95$ in all models; see Table~\ref{tab:res}) include Canada $[0.980,0.984]$, the United States $[0.973,0.980]$, Finland $[0.972,0.976]$, Lithuania $[0.971,0.974]$, Hungary $[0.969,0.975]$, Australia $[0.955,0.966]$, Slovenia $[0.951,0.959]$, and the Czech Republic $[0.950,0.963]$. Countries reaching $\rho^* \geq 0.95$ only in some models are Denmark $[0.942,0.950]$, Portugal $[0.909,0.950]$, and Italy $[0.908,0.951]$. Particularly, Portugal and Italy only reach the $0.950$ threshold when accounting for ESCS and in certain directions.

In the radial case, some differences emerge when ESCS is incorporated as an input compared to when it is excluded. A notable improvement in the efficiency score $\rho^*$ is observed for Colombia $(+0.060)$, which becomes efficient. 
More modest changes are identified for Portugal $(+0.020)$, Germany $(+0.012)$, Austria $(+0.010)$, Italy $(+0.009)$, the Czech Republic $(+0.007)$, Belgium $(+0.004)$, and Spain $(+0.002)$. The remaining countries exhibit no significative changes in their efficiency scores.

Although the efficiency scores $\rho^*$ obtained from radial and directional models are very similar, directional models determine targets in accordance with the improvement strategy specified by the orientation coefficients. Consequently, the targets derived from directional models  may differ substantially from those obtained using radial models (see Tables~\ref{tab:targetsvrs13}, \ref{tab:targetsvrs134}). This distinction may be critical in establishing the performance levels that each country should attain in order to achieve efficiency (see Section \ref{sec:models} for an example with Spain).

\subsubsection{Adaptive constrained enveloping splines}
\label{sec:res_aces}

To complement the classical radial and directional analysis presented above, we re-evaluate the same dataset using an ACES-based technology. The aim is not to replace the deterministic DEA results, but to assess how efficiency scores change when the frontier is estimated through a regularized ML procedure designed to reduce the overfitting problems of conventional DEA. 

The final ACES specification depends on several hyperparameters controlling the complexity of the basis expansion, the admissible interaction structure, the pruning intensity, the shape restrictions, the LOF criterion used in the forward stage, and the refinement step. In this application, not all hyperparameters were tuned jointly. Some were fixed in advance, while the remaining ones were selected through a grid-search procedure adapted to frontier estimation.

The grid search was conducted over the interaction-admission threshold $\xi^{(q)} \in \{0, 0.01, 0.025, 0.05\}$, the decision of whether to impose monotonicity and/or concavity in model~\eqref{eq:aces_multioutput_lp}, the pruning penalty $\gamma \in \{1, 2\}$ in \eqref{eq:gcv}, and the refinement threshold $\psi \in \{0, 0.01, 0.025, 0.05, 0.10\}$ in \eqref{eq:refinement}. The remaining arguments were fixed throughout the exercise. Specifically, the maximum interaction degree $q_{\max}$ was set to $3$ in the three-input specification and to $2$ in the two-input specification; the LOF criterion was fixed to the MAE; the minimum relative improvement threshold $\xi$ was fixed at $0.005$; and all observed data points were considered as candidate knots, with a minimum span of two observations between consecutive knots. Finally, the upper bound on the number of BF, $T_{\max}$, was set equal to the size of the training sample.

Hyperparameter selection was carried out through a leave-one-out cross-validation (LOOCV) procedure specifically tailored to frontier estimation. Since standard CV may favour models that fit inefficient interior observations too closely, the validation criterion was restricted to countries located near the empirical frontier. To identify these units, a preliminary output-oriented radial efficiency score $\eta_j^*$ was computed for each country using the standard VRS DEA model in \eqref{eq:rad}. During the LOOCV procedure, the prediction error of the excluded country was included in the validation criterion only if $\eta_j^* < 1.025$, meaning that only countries lying at most $2.5\%$ away from the deterministic frontier contributed to the tuning stage. In addition, inverse-efficiency weights $\omega_j=1/\eta_j^*$ were used, giving relatively greater importance to countries closer to the empirical frontier. This criterion retained ten countries in the three-input specification and nine countries in the two-input specification. For some of these units, however, no valid out-of-sample target could be obtained after excluding the evaluated country from the reference sample. Since the same folds were unavailable for every candidate hyperparameter configuration, model comparison was performed on the common set of six valid folds, corresponding to Canada, Ireland, Japan, Korea, Latvia, and the United States.

For a configuration $h$, let $\mathcal{K}$ denote the common set of validation folds corresponding to countries satisfying $\eta_j^* < 1.025$ for which a valid out-of-sample target can be obtained. The fold-specific error across the $s=3$ outputs is evaluated as $\mathrm{RMSE}_k(h) = \sqrt{\frac{1}{s}\sum_{r=1}^{s}\left(\hat{y}^{h}_{j_k r}-y_{j_k r}\right)^2}$. Candidate specifications are ranked by their median error, $\mathrm{CV}(h) = \operatorname{median}_{k\in\mathcal{K}}\left\{\mathrm{RMSE}_k(h)\right\}$, selecting $h^* = \arg\min_{h\in\mathcal{H}}\mathrm{CV}(h)$. The same set $\mathcal{K}$ is used for all candidate configurations, ensuring that their validation errors are directly comparable. The median is preferred over the mean because the trimmed validation sample is small, making averages highly sensitive to outliers. Since the goal is to recover the technological boundary rather than the conditional mean, this robustness to extreme errors provides a more reliable criterion for out-of-sample performance near the frontier.

\begin{table}[htbp]
\centering
\caption{Final ACES hyperparameter configuration selected by median-based LOOCV.}
\label{tab:aces_hyperparameters}
\begin{tabular}{lcc}
\hline
Hyperparameter & Three-input specification & Two-input specification \\
\hline
$q_{\text{max}}$ & 3 & 2 \\
$\xi^{(q)}$ & 0.010 & 0.000 \\
Monotonicity & Yes & Yes \\
Concavity & No & No \\
$\xi$ & 0.005 & 0.005 \\
minspan & 2 & 2 \\
$\gamma$ & 1 & 1 \\
$\psi$ & 0.025 & 0.025 \\
LOF criterion & MAE & MAE \\
\hline
\end{tabular}
\end{table}

Using the configurations reported in Table~\ref{tab:aces_hyperparameters}, the ACES results reveal that sensitivity to frontier estimation varies substantially across countries. Many efficiency scores remain unchanged at the reported precision, while others indicate additional scope for output improvement (Tables~\ref{tab:res} and \ref{tab:res_aces}). The direction of this adjustment follows from the construction described in Section~\ref{sec:aces_background}: the refined predictions are no smaller than the observed outputs, so the resulting technology contains the classical VRS production possibility set. For a given input specification and improvement direction, efficiency scores therefore cannot increase. The informative result is the magnitude and distribution of the adjustments, which identify where the assessment depends most strongly on how the frontier is estimated.

This distinction matters particularly when interpreting frontier membership. Ireland and Japan retain their frontier positions across the specifications considered, supporting the stability of their favourable assessment. Korea also remains very close to the frontier, although its reported score decreases from $1.000$ to $0.999$. Its revised targets differ from observed performance by less than one PISA point in every subject and direction (Tables~\ref{tab:targetsvrs13_ACES} and \ref{tab:targetsvrs134_ACES}). In this case, a change in the efficiency classification accompanies only a small change in the estimated improvement requirement. Interpreting frontier membership alongside the corresponding targets avoids assigning disproportionate substantive importance to such marginal differences.

Larger revisions for countries already classified as inefficient have more direct implications for improvement planning. Finland illustrates this distinction: without ESCS, its radial targets increase by approximately eleven PISA points in each subject, while some directional targets increase by about seventeen points (Tables~\ref{tab:targetsvrs13} and \ref{tab:targetsvrs13_ACES}). Its classification remains unchanged, but the estimated distance to the boundary becomes appreciably larger. The increase appears under both radial and directional specifications, indicating that it is not confined to the assumption of proportional output expansion. The direction still determines how the improvement effort is distributed across subjects, while frontier estimation determines how far that direction can extend within the estimated technology. Both choices therefore matter when translating efficiency analysis into output targets.

The inclusion of ESCS introduces a further source of sensitivity. The differences between classical DEA and ACES do not become uniformly smaller once socioeconomic background is incorporated. Colombia provides a particularly clear example: the unit efficiency obtained under the three-input specification falls to $0.940$ in the corresponding radial ACES model (Tables~\ref{tab:res} and \ref{tab:res_aces}). Its favourable reassessment after introducing ESCS in classical DEA is therefore sensitive to the frontier estimator. Moreover, because ACES is fitted separately for each input specification, adding ESCS changes both the restrictions on reference combinations and the predicted outputs used to construct the technology.

These findings give ACES a specific role in the interpretation of the educational efficiency results. Agreement between the two approaches identifies assessments that persist under the selected alternative frontier specification, whereas larger differences identify countries whose estimated improvement requirements deserve closer examination. Lower efficiency scores alone do not establish that ACES has recovered the true production frontier or that its targets are achievable through a particular policy. Its contribution here is to make the dependence of the benchmarking conclusions on frontier estimation explicit, including cases where apparently stable efficiency classifications conceal meaningful changes in the output targets.

\subsubsection{Models with imprecise data}
\label{sec:impreciseres}

As in the previous analyses, we consider two scenarios. In the first scenario, we take instruction time and expenditure as inputs. In the second scenario, we also consider ESCS as an input. 

The implementation of both radial \eqref{eq:radcc} and directional \eqref{eq:dirdcc} chance-constrained models across the two scenarios results in all countries being classified as stochastically efficient for $\alpha \leq 0.25$, providing no discrimination among countries. This result may be attributed to the high variance of the stochastic variables (see Table~\ref{tab:data}), which means that inefficient scores are obtained only for relatively high values of $\alpha$. 

Table~\ref{tab:cceff} reports the efficiency scores $\rho^*$ obtained from the chance-constrained radial models for values of $\alpha$ between $0.3$ and $0.5$ under both scenarios. For $\alpha=0.3$, Iceland is the only inefficient country, while Greece also becomes inefficient for $\alpha=0.35$. For $\alpha=0.4$, Israel, the Netherlands, and Norway are inefficient under both scenarios, whereas Belgium, France, Germany, Italy, Portugal, the Slovak Republic, and Spain are inefficient only under the first scenario. For $\alpha=0.45$, only Canada, Finland, Hungary, Lithuania, and the United States are efficient in the first scenario, in addition to the countries that are efficient in the deterministic case, namely Ireland, Japan, Korea, Latvia, Mexico, Poland, and T\"urkiye. Under the second scenario, the Czech Republic is also efficient, together with the countries that are efficient in the deterministic case, which in this scenario also include Colombia. Finally, the case $\alpha=0.5$ coincides with the deterministic case.

We also compute fuzzy efficiency scores for the OECD countries under analysis by applying the Kao--Liu fuzzy meta-model with radial model, given by \eqref{eq:radworst} and \eqref{eq:radbest}. For this purpose, a fuzzy number is constructed for each performance using percentiles, as illustrated in Figure~\ref{fig:fuzzy}. The adopted criterion is that each $\alpha$-cut contains $100(1-\alpha)\%$ of the sample, with bounds defined by symmetric percentiles of the form $\textbf{P}_x$ and $\textbf{P}_{100-x}$, where $x = 50\alpha$.
Percentiles are provided by the PISA 2022 report \citep{pisa2022} for mathematics (Table I.B1.2.1.), reading (Table I.B1.2.2.) and science (Table I.B1.2.3.) performances, under the name ``Percentiles''.
With the published data (see Table~\ref{tab:fuzzyper}), we can determine $\alpha$-cuts with $\alpha =0.2$ (whose extremes are $\textbf{P}_{10}$ and $\textbf{P}_{90}$, containing $80\%$ of the sample), $\alpha =0.5$ (whose extremes are $\textbf{P}_{25}$ and $\textbf{P}_{75}$, containing $50\%$ of the sample), and $\alpha =1$ (which is crisp, given by the median $\textbf{P}_{50}$). In order to determine the $0$-cuts, we would need the lowest and highest individual scores for each performance, but this information is not published. Moreover, new $\alpha$-cuts for $\alpha$ between $0.2$ and $1$ can be created by interpolating. However, it should be noted that these new $\alpha$-cuts do not have the probabilistic interpretation that the originals do, which are constructed using percentiles. Finally, a Kao--Liu meta-model is applied to these $\alpha$-cuts.

\begin{figure}[htbp]
\centering
  \includegraphics[width=0.3\textwidth]{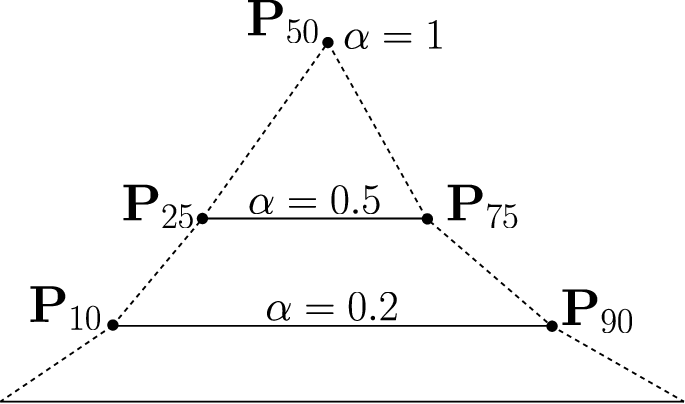}
\caption{Construction of a fuzzy number from percentiles.}
\label{fig:fuzzy}
\end{figure}

For the first scenario, the resulting fuzzy efficiency scores $\rho^*$ are presented in Table~\ref{tab:fuzzy_12} and Figure~\ref{fig:fuzzy12}. The membership functions of the fuzzy efficiencies for Mexico, Poland, and T\"urkiye degenerate to a crisp value equal to $1$. Furthermore, the associated inefficiency slacks are zero, confirming that they are therefore efficient countries. 
In other words, these three countries consistently lie on the best-practice frontier, even when the production frontier shifts due to variations in the inputs and outputs considered.

The $\alpha $-cuts provide the intervals of the fuzzy efficiency score at different possibility levels. Thus, for the lowest $\alpha $-cut available ($\alpha =0.2$), the range of the efficiency score is obtained, discarding the top $10\%$ and bottom $10\%$ of the performance scores. For example, in the first scenario and for the case of Spain, the interval of the efficiency score is $\left[ 0.560, 1\right] $, referring to Table~\ref{tab:fuzzy_12}. This result indicates that the efficiency score for Spain (considering the central $80\%$ of the performance scores), although imprecise, will fall between $0.560$ and $1$. It is noteworthy that Israel is the country exhibiting the greatest uncertainty in efficiency scores, as reflected by the widest fuzzy efficiency interval: $\left[ 0.493 , 1\right] $. At the other extreme, the $\alpha $-cut for $\alpha =1$ indicates the most likely efficiency score. At $\alpha =1$ the best-practice frontier is formed by Ireland, Japan, Korea, Latvia, Mexico, Poland, and T\"urkiye; the remaining countries obtain efficiency scores below $1$ and are therefore classified as inefficient. Note that for this $\alpha $-cut, the efficiency scores are crisp (the lower and upper scores are the same); this is because the top $\alpha $-cut (with $\alpha =1$) is the value $\textbf{P}_{50}$ and not an interval.

Similar observations apply to the second scenario (see Table~\ref{tab:fuzzy_123} and Figure~\ref{fig:fuzzy123}). However, in this case, Colombia becomes efficient in all the $\alpha$-cuts. On the other hand, Israel continues exhibiting the greatest uncertainty in its efficiency score. 
Consistent with the results obtained from the classical models, some improvements are observed for Portugal, Germany, Austria, Italy, the Czech Republic, Belgium, and Spain.

\section{Conclusions}
\label{sec:conclusions}

This paper develops a structured DEA framework for evaluating the efficiency of countries based on PISA outcomes, using mean performance scores as outputs. The main theoretical contributions are Proposition \ref{prop:ratio} and Theorem \ref{thm}, which establish the conditions under which common-denominator ratio variables, such as mean performance scores, can be appropriately incorporated into standard VRS models. Specifically, Theorem \ref{thm} shows that a CRS problem with a non-controllable variable can be reformulated as a VRS problem by normalizing all inputs and outputs with respect to the non-controllable variable. These results are particularly useful, for instance, in conducting archetypal analysis of DMUs, where the VRS assumption is required  \citep{alcacer2025survey}.

Within the PISA application, the main methodological contributions consist of several systematic approaches for estimating improvement direction vectors in directional models. Combined with the inclusion of the ESCS as an input variable, these approaches enable the computation of customized and realistic targets, reflecting the heterogeneous improvement possibilities across countries and providing a more equitable basis for cross-country comparisons.

Classical radial and directional DEA models provide the main framework for the efficiency assessment conducted in this study. Their flexibility and interpretability make them suitable for applications involving educational performance indicators such as PISA scores. The proposed enhancements build directly upon these models without compromising their interpretability, reinforcing their central role in empirical applications.

In addition, several complementary methodologies are considered to provide different perspectives on the efficiency analysis. These approaches are regarded as extensions rather than substitutes for classical DEA models, as each involves specific advantages and limitations. In particular, the ACES methodology can help mitigate potential overestimation of efficiency, especially under variable returns to scale. Given the relatively recent development of ACES, its application remains limited in the literature. In this study, we extend its application by introducing several methodological developments, including the computation of targets on the estimated frontier based on the direction of improvement. Moreover, because the implementation of ACES requires user-defined parameters, which may affect the results depending on the degree of strictness imposed, we propose a parameter-selection procedure tailored to the PISA educational efficiency context. To facilitate the application of the methodology, we also provide scripts that allow non-expert users to replicate the results.

Stochastic chance-constrained and fuzzy DEA models are also considered. In the present application, chance-constrained approaches provide limited additional insights, partly because of the relatively high variability observed in the data. Fuzzy models, in contrast, provide an alternative way of incorporating uncertainty into the analysis and obtaining more flexible efficiency estimates. Nevertheless, they do not fully capture the underlying stochastic nature of the variables, and their interpretation may be less straightforward.

Overall, the framework presented in this study can be readily extended to other datasets, including past and future PISA reports or similar large-scale educational assessments, even in other areas besides education. Therefore, this methodology represents a useful tool for researchers and practitioners seeking to conduct fair and informative efficiency analyses in educational or other contexts.

\clearpage

\appendix

\setcounter{table}{0}
\counterwithin{table}{section}

\counterwithin{figure}{section}
\renewcommand{\thefigure}{\thesection.\arabic{figure}}

\section{Data and results tables}\label{appendix}

\begin{sidewaystable}[ht]
\centering
\footnotesize
\begin{tabular}{l|rrlr|rrrrrr}
              & Instruction & Expenditure & ESCS   & ESCS SD & Mathematics & Math. SD & Reading & Read. SD & Science & Sci. SD \\ \hline
Australia     & 110.00      & 126.02      & 356.48 & 84.46   & 487.08      & 99.26    & 498.05  & 110.94   & 507.00  & 108.93  \\
Austria       & 65.40       & 154.07      & 325.45 & 94.30   & 487.27      & 93.58    & 480.41  & 104.13   & 491.27  & 101.11  \\
Belgium       & 67.85       & 130.34      & 326.83 & 92.53   & 489.49      & 96.33    & 478.85  & 104.80   & 490.58  & 100.98  \\
Canada        & 83.07       & 121.68      & 356.65 & 75.55   & 496.95      & 94.01    & 507.13  & 108.72   & 515.02  & 100.57  \\
Chile         & 83.00       & 60.63       & 267.68 & 94.13   & 411.70      & 76.63    & 447.98  & 93.27    & 443.54  & 91.82   \\
Colombia      & 98.00       & 37.31       & 212.03 & 120.20  & 382.70      & 72.82    & 408.67  & 93.25    & 411.12  & 86.83   \\
Czech Rep.    & 69.87       & 100.84      & 308.14 & 86.75   & 487.00      & 93.35    & 488.60  & 97.94    & 497.74  & 99.49   \\
Denmark       & 106.00      & 130.69      & 366.15 & 74.84   & 489.27      & 81.56    & 488.80  & 91.90    & 493.82  & 94.70   \\
Finland       & 63.84       & 126.78      & 344.39 & 81.29   & 484.14      & 89.28    & 490.22  & 104.03   & 510.96  & 106.38  \\
France        & 81.92       & 109.58      & 318.56 & 92.28   & 473.94      & 91.06    & 473.85  & 105.84   & 487.23  & 102.99  \\
Germany       & 73.76       & 121.06      & 304.85 & 104.25  & 474.83      & 94.71    & 479.79  & 105.91   & 492.43  & 106.38  \\
Greece        & 69.15       & 71.51       & 303.06 & 92.36   & 430.15      & 83.37    & 438.44  & 94.32    & 440.79  & 90.84   \\
Hungary       & 59.28       & 78.97       & 318.84 & 95.63   & 472.78      & 93.90    & 472.97  & 101.46   & 485.89  & 96.45   \\
Iceland       & 76.20       & 149.64      & 356.34 & 78.25   & 458.90      & 87.78    & 435.90  & 103.03   & 446.93  & 94.84   \\
Ireland       & 81.72       & 94.17       & 351.84 & 80.39   & 491.65      & 79.62    & 516.01  & 88.18    & 503.85  & 91.26   \\
Israel        & 84.36       & 94.35       & 346.54 & 92.39   & 457.90      & 107.07   & 473.83  & 122.12   & 464.75  & 109.13  \\
Italy         & 74.90       & 105.75      & 308.52 & 93.06   & 471.26      & 88.90    & 481.60  & 92.27    & 477.46  & 92.75   \\
Japan         & 73.38       & 101.40      & 317.86 & 71.33   & 535.58      & 92.75    & 515.85  & 96.26    & 546.63  & 92.97   \\
Korea         & 64.56       & 144.48      & 340.98 & 82.35   & 527.30      & 105.20   & 515.42  & 103.28   & 527.82  & 105.49  \\
Latvia        & 58.38       & 68.78       & 317.60 & 83.03   & 483.16      & 80.15    & 474.57  & 89.65    & 493.84  & 84.89   \\
Lithuania     & 76.28       & 72.11       & 323.77 & 88.84   & 475.15      & 87.21    & 471.83  & 94.24    & 484.46  & 92.39   \\
Mexico        & 78.85       & 28.90       & 223.30 & 116.35  & 395.03      & 69.39    & 415.36  & 84.33    & 409.89  & 74.90   \\
Netherlands   & 86.40       & 119.58      & 343.89 & 87.32   & 492.68      & 106.09   & 459.24  & 114.88   & 488.32  & 112.26  \\
Norway        & 78.93       & 153.34      & 370.69 & 83.45   & 468.45      & 93.47    & 476.52  & 112.39   & 478.23  & 105.99  \\
Poland        & 52.48       & 87.74       & 308.03 & 88.61   & 488.96      & 89.46    & 488.71  & 104.00   & 499.16  & 96.19   \\
Portugal      & 76.98       & 98.98       & 295.91 & 113.77  & 471.91      & 89.63    & 476.59  & 93.54    & 484.37  & 92.04   \\
Slovak Rep.   & 68.23       & 75.30       & 288.24 & 96.29   & 463.99      & 101.09   & 446.86  & 104.72   & 462.27  & 103.31  \\
Slovenia      & 63.90       & 102.34      & 341.13 & 84.38   & 484.53      & 89.23    & 468.54  & 96.54    & 499.96  & 94.20   \\
Spain         & 79.23       & 93.09       & 315.22 & 100.92  & 473.14      & 86.37    & 474.31  & 96.73    & 484.53  & 91.66   \\
Sweden        & 68.91       & 133.02      & 351.20 & 84.83   & 481.77      & 95.60    & 486.98  & 110.77   & 493.55  & 108.28  \\
Türkiye       & 62.52       & 46.71       & 200.00 & 117.19  & 453.15      & 89.81    & 456.08  & 86.53    & 475.94  & 89.37   \\
United States & 89.13       & 143.38      & 324.16 & 97.99   & 464.89      & 94.51    & 503.94  & 111.44   & 499.41  & 108.13 
\end{tabular}
\caption{Data for variables of the OECD countries considered as DMUs. This includes instruction time in general education ($100$ hours unit), expenditure on educational institutions per student ($1\,000$ dollars unit), and the index of economic social and cultural status (ESCS) (shifted so that its minimum value is $2$ and then multiplied by $100$). The outputs are the scores of mathematics, reading and science performances. Standard deviations for the ESCS and scores are also provided for the stochastic analysis.}
\label{tab:data}
\end{sidewaystable}

\setlength{\tabcolsep}{3pt}
\begin{sidewaystable}[ht]
\centering
\footnotesize
\begin{tabular}{l|rrr|rrr|rrr|rrr|rrr|rrr}
                & \multicolumn{3}{c|}{Rel. ann. rat. change (\%)}                                    & \multicolumn{3}{c|}{Orient. by (18) 
$\delta =0\%$} & \multicolumn{3}{c|}{Orient. by (18) 
$\delta =1\%$} & \multicolumn{3}{c|}{Orient. by (19)
} & \multicolumn{3}{c}{Orient. by (20) 
$\delta =0\%$} & \multicolumn{3}{c}{Orient. by (20) 
$\delta =5\%$} \\
                & \multicolumn{1}{c}{Math.} & \multicolumn{1}{c}{Read.} & \multicolumn{1}{c|}{Sci.} & \multicolumn{1}{c}{$d_1$ (\%)} & \multicolumn{1}{c}{$d_2$ (\%)} & \multicolumn{1}{c|}{$d_3$ (\%)} & \multicolumn{1}{c}{$d_1$ (\%)} & \multicolumn{1}{c}{$d_2$ (\%)} & \multicolumn{1}{c|}{$d_3$ (\%)} & \multicolumn{1}{c}{$d_1$ (\%)} & \multicolumn{1}{c}{$d_2$ (\%)} & \multicolumn{1}{c|}{$d_3$ (\%)} & \multicolumn{1}{c}{$d_1$ (\%)} & \multicolumn{1}{c}{$d_2$ (\%)} & \multicolumn{1}{c|}{$d_3$ (\%)} & \multicolumn{1}{c}{$d_1$ (\%)} & \multicolumn{1}{c}{$d_2$ (\%)} & \multicolumn{1}{c}{$d_3$ (\%)} \\ \hline
Australia       & -0.302                         & -0.169                     & -0.159                     & 1.251          & 1.384          & 1.394          & 2.251          & 2.384          & 2.394          & 27.76  & 27.11  & 27.36  & 7.425           & 2.607           & 4.681           & 12.425          & 7.607           & 9.681           \\
Austria         & -0.461                         & -0.202                     & -0.157                     & 1.092          & 1.351          & 1.396          & 2.092          & 2.351          & 2.396          & 24.43  & 26.35  & 25.65  & 4.199           & 2.025           & 3.208           & 9.199           & 7.025           & 8.208           \\
\end{tabular}
\caption{Relative annual rates of change in 2022 of mathematics, reading and science performances, jointly with orientation coefficients computed from different methods.}
\label{tab:dir}
\end{sidewaystable}

\setlength{\tabcolsep}{6pt}
\begin{sidewaystable}[ht]
\centering
\footnotesize
\begin{tabular}{l|rrrrr|rrrrr|rrrrr}
\multicolumn{1}{c|}{} & \multicolumn{5}{c|}{Mathematics}                                                                                                 & \multicolumn{5}{c|}{Reading}                                                                                                     & \multicolumn{5}{c}{Science}                                                                                                     \\
\multicolumn{1}{c|}{} & \multicolumn{1}{c}{$\textbf{P}_{10}$} & \multicolumn{1}{c}{$\textbf{P}_{25}$} & \multicolumn{1}{c}{$\textbf{P}_{50}$} & \multicolumn{1}{c}{$\textbf{P}_{75}$} & \multicolumn{1}{c|}{$\textbf{P}_{90}$} & \multicolumn{1}{c}{$\textbf{P}_{10}$} & \multicolumn{1}{c}{$\textbf{P}_{25}$} & \multicolumn{1}{c}{$\textbf{P}_{50}$} & \multicolumn{1}{c}{$\textbf{P}_{75}$} & \multicolumn{1}{c|}{$\textbf{P}_{90}$} & \multicolumn{1}{c}{$\textbf{P}_{10}$} & \multicolumn{1}{c}{$\textbf{P}_{25}$} & \multicolumn{1}{c}{$\textbf{P}_{50}$} & \multicolumn{1}{c}{$\textbf{P}_{75}$} & \multicolumn{1}{c}{$\textbf{P}_{90}$} \\ \hline
Australia             & 358.47                  & 416.10                  & 484.56                  & 555.94                  & 619.09                   & 350.59                  & 422.38                  & 502.19                  & 576.25                  & 638.34                   & 364.19                  & 430.42                  & 508.37                  & 583.49                  & 647.46                  \\
Austria               & 362.26                  & 420.42                  & 488.97                  & 554.38                  & 608.43                   & 339.53                  & 406.44                  & 485.12                  & 556.57                  & 612.97                   & 355.56                  & 417.94                  & 494.67                  & 565.04                  & 621.53                  \\
\end{tabular}
\caption{Percentiles for each performance to construct fuzzy numbers, as they were published in \citep[Table I.B1.2.1]{pisa2022}, \citep[Table I.B1.2.2]{pisa2022} and \citep[Table I.B1.2.3]{pisa2022}.}
\label{tab:fuzzyper}
\end{sidewaystable}

\begin{table}[ht]
\centering
\footnotesize
\begin{tabular}{l|rrrr|rrrr}
& \multicolumn{4}{|c|}{No ESCS} & \multicolumn{4}{c}{ESCS} \\
& Radial & Dir. (18)
& Dir. (19)
& Dir. (20)
& Radial & Dir. (18)
& Dir. (19)
& Dir. (20)
\\
\hline
Australia             & 0.965          & 0.966          & 0.965          & 0.955          & 0.965          & 0.966          & 0.965          & 0.955          \\
Austria               & 0.932          & 0.934          & 0.934          & 0.931          & 0.942          & 0.944          & 0.944          & 0.940          \\
Belgium               & 0.929          & 0.924          & 0.931          & 0.935          & 0.933          & 0.928          & 0.935          & 0.938          \\
Canada                & 0.983          & 0.983          & 0.984          & 0.980          & 0.983          & 0.983          & 0.984          & 0.980          \\
Chile                 & 0.946          & 0.946          & 0.946          & 0.940          & 0.946          & 0.946          & 0.946          & 0.940          \\
Colombia              & 0.940          & 0.940          & 0.945          & 0.939          & \textbf{1.000}          & \textbf{1.000}          & \textbf{1.000}          & \textbf{1.000}          \\
\end{tabular}
\caption{Efficiency scores $\rho ^*$ without considering and considering ESCS as an input for different radial and directional models with deterministic variables under VRS. Threshold $\delta$ is taken $1\%$ for directions (18) 
and $5\%$ for directions (20)
.}
\label{tab:res}
\end{table}   

\begin{table}[ht]
\centering
\footnotesize
\begin{tabular}{l|rrrr|rrrr}
& \multicolumn{4}{|c|}{No ESCS} & \multicolumn{4}{c}{ESCS} \\
& Radial & Dir. (18)
& Dir. (19)
& Dir. (20)
& Radial & Dir. (18)
& Dir. (19)
& Dir. (20)
\\
\hline
Australia       & 0.965 & 0.966 & 0.965 & 0.955 & 0.965 & 0.966 & 0.965 & 0.955 \\
Austria         & 0.931 & 0.933 & 0.933 & 0.921 & 0.931 & 0.933 & 0.933 & 0.921 \\
Belgium         & 0.928 & 0.922 & 0.930 & 0.927 & 0.928 & 0.922 & 0.931 & 0.927 \\
Canada          & 0.983 & 0.983 & 0.983 & 0.980 & 0.983 & 0.983 & 0.983 & 0.980 \\
Chile           & 0.946 & 0.946 & 0.946 & 0.940 & 0.924 & 0.924 & 0.924 & 0.916 \\
Colombia        & 0.940 & 0.931 & 0.945 & 0.939 & 0.940 & 0.940 & 0.945 & 0.939 \\
\end{tabular}
\caption{Efficiency scores $\rho^*$ using ACES, without considering and considering ESCS as an input for different radial and directional models with deterministic variables under VRS. Threshold $\delta$ is taken as $1\%$ for direction (18) and $5\%$ for direction (20).}
\label{tab:res_aces}
\end{table} 

\begin{sidewaystable}[ht]
\centering
\footnotesize
\begin{tabular}{l|rrr|rrr|rrr|rrr}
\multicolumn{1}{c}{}    & \multicolumn{3}{|c|}{Radial} & \multicolumn{3}{c|}{Dir. (18)
} & \multicolumn{3}{c|}{Dir. (19)
} & \multicolumn{3}{c}{Dir. (20)
} \\
& \multicolumn{1}{c}{Math.} & \multicolumn{1}{c}{Read.} & \multicolumn{1}{c|}{Sci.} & \multicolumn{1}{c}{Math.} & \multicolumn{1}{c}{Read.} & \multicolumn{1}{c|}{Sci.} & \multicolumn{1}{c}{Math.} & \multicolumn{1}{c}{Read.} & \multicolumn{1}{c|}{Sci.} & \multicolumn{1}{c}{Math.} & \multicolumn{1}{c}{Read.} & \multicolumn{1}{c}{Sci.} \\ \hline
Australia             & 504.57                    & 515.93                    & 525.20                    & 503.59                    & 515.93                    & 525.28                    & 504.99                    & 515.93                    & 525.37                    & 515.62                    & 515.91                    & 530.14                   \\
Austria               & 522.82                    & 515.46                    & 527.12                    & 518.90                    & 515.46                    & 527.80                    & 520.22                    & 515.46                    & 526.15                    & 528.09                    & 511.14                    & 527.99                   \\
Belgium               & 527.03                    & 515.58                    & 528.20                    & 530.39                    & 514.89                    & 533.93                    & 525.05                    & 515.58                    & 526.17                    & 530.39                    & 508.64                    & 521.40                   \\
\end{tabular}
\caption{Targets in performance scores on the boundary of the PPS ($\partial ^{\text{W}}P$) considering instruction time and expenditure as inputs, computed by radial and directional models with deterministic variables under VRS. Threshold $\delta$ is taken $1\%$ for directions (18) 
and $5\%$ for directions (20)
.}
\label{tab:targetsvrs13}
\end{sidewaystable}

\begin{sidewaystable}[ht]
\centering
\footnotesize
\begin{tabular}{l|rrr|rrr|rrr|rrr}
\multicolumn{1}{c}{}    & \multicolumn{3}{|c|}{Radial} & \multicolumn{3}{c|}{Dir. (18)
} & \multicolumn{3}{c|}{Dir. (19)
} & \multicolumn{3}{c}{Dir. (20)
} \\
& \multicolumn{1}{c}{Math.} & \multicolumn{1}{c}{Read.} & \multicolumn{1}{c|}{Sci.} 
& \multicolumn{1}{c}{Math.} & \multicolumn{1}{c}{Read.} & \multicolumn{1}{c|}{Sci.} 
& \multicolumn{1}{c}{Math.} & \multicolumn{1}{c}{Read.} & \multicolumn{1}{c|}{Sci.} 
& \multicolumn{1}{c}{Math.} & \multicolumn{1}{c}{Read.} & \multicolumn{1}{c}{Sci.} \\ 
\hline
Australia       & 504.65 & 516.01 & 525.28 & 503.67 & 516.01 & 525.36 & 505.07 & 516.01 & 525.45 & 515.77 & 516.01 & 530.27 \\
Austria         & 523.38 & 516.01 & 527.68 & 519.39 & 516.01 & 528.38 & 520.74 & 516.01 & 526.70 & 534.56 & 516.01 & 533.81 \\
Belgium         & 527.47 & 516.01 & 528.65 & 531.65 & 516.01 & 535.27 & 525.47 & 516.01 & 526.59 & 535.58 & 512.42 & 525.32 \\
Canada          & 505.65 & 516.01 & 524.03 & 504.81 & 516.01 & 524.31 & 505.22 & 516.01 & 523.54 & 508.64 & 516.01 & 525.43 \\
Chile           & 435.30 & 473.66 & 468.97 & 432.47 & 473.66 & 471.65 & 434.43 & 473.66 & 469.58 & 441.48 & 473.66 & 471.08 \\
Colombia        & 406.98 & 434.60 & 437.21 & 412.46 & 433.12 & 445.44 & 403.28 & 434.60 & 435.36 & 408.27 & 434.60 & 437.88 \\
\end{tabular}
\caption{Targets in performance scores on the boundary of the PPS ($\partial^{\text{W}}P$) considering instruction time and expenditure as inputs, computed by radial and directional models using ACES under VRS. Threshold $\delta$ is taken as $1\%$ for direction (18) and $5\%$ for direction (20).}
\label{tab:targetsvrs13_ACES}
\end{sidewaystable}

\begin{sidewaystable}[ht]
\centering
\footnotesize
\begin{tabular}{l|rrr|rrr|rrr|rrr}
\multicolumn{1}{c}{}    & \multicolumn{3}{|c|}{Radial} & \multicolumn{3}{c|}{Dir. (18)
} & \multicolumn{3}{c|}{Dir. (19)
} & \multicolumn{3}{c}{Dir. (20)
} \\
& \multicolumn{1}{c}{Math.} & \multicolumn{1}{c}{Read.} & \multicolumn{1}{c|}{Sci.} & \multicolumn{1}{c}{Math.} & \multicolumn{1}{c}{Read.} & \multicolumn{1}{c|}{Sci.} & \multicolumn{1}{c}{Math.} & \multicolumn{1}{c}{Read.} & \multicolumn{1}{c|}{Sci.} & \multicolumn{1}{c}{Math.} & \multicolumn{1}{c}{Read.} & \multicolumn{1}{c}{Sci.} \\ \hline
Australia             & 504.57                    & 515.93                    & 525.20                    & 503.59                    & 515.93                    & 525.28                    & 504.99                    & 515.93                    & 525.37                    & 515.62                    & 515.91                    & 530.14                   \\
Austria               & 517.36                    & 510.07                    & 521.61                    & 514.04                    & 510.07                    & 522.19                    & 515.16                    & 510.07                    & 520.80                    & 522.52                    & 506.95                    & 522.98                   \\
Belgium               & 524.75                    & 513.34                    & 525.92                    & 528.08                    & 512.86                    & 531.48                    & 522.89                    & 513.34                    & 524.01                    & 528.08                    & 506.95                    & 519.66                   \\
\end{tabular}
\caption{Targets in performance scores on the boundary of the PPS ($\partial ^{\text{W}}P$) considering instruction time, expenditure and ESCS as inputs, computed by radial and directional models with deterministic variables under VRS. Threshold $\delta$ is taken $1\%$ for directions (18) 
and $5\%$ for directions (20)
.}
\label{tab:targetsvrs134}
\end{sidewaystable}

\begin{sidewaystable}[ht]
\centering
\footnotesize
\begin{tabular}{l|rrr|rrr|rrr|rrr}
\multicolumn{1}{c}{}    & \multicolumn{3}{|c|}{Radial} & \multicolumn{3}{c|}{Dir. (18)
} & \multicolumn{3}{c|}{Dir. (19)
} & \multicolumn{3}{c}{Dir. (20)
} \\
& \multicolumn{1}{c}{Math.} & \multicolumn{1}{c}{Read.} & \multicolumn{1}{c|}{Sci.} 
& \multicolumn{1}{c}{Math.} & \multicolumn{1}{c}{Read.} & \multicolumn{1}{c|}{Sci.} 
& \multicolumn{1}{c}{Math.} & \multicolumn{1}{c}{Read.} & \multicolumn{1}{c|}{Sci.} 
& \multicolumn{1}{c}{Math.} & \multicolumn{1}{c}{Read.} & \multicolumn{1}{c}{Sci.} \\ 
\hline
Australia       & 504.65 & 516.01 & 525.28 & 503.66 & 516.01 & 525.35 & 505.07 & 516.01 & 525.45 & 515.73 & 515.98 & 530.23 \\
Austria         & 523.26 & 515.89 & 527.56 & 519.29 & 515.89 & 528.25 & 520.63 & 515.89 & 526.58 & 534.40 & 515.89 & 533.67 \\
Belgium         & 527.35 & 515.90 & 528.53 & 531.52 & 515.90 & 535.13 & 525.36 & 515.90 & 526.48 & 535.58 & 512.42 & 525.32 \\
Canada          & 505.64 & 516.00 & 524.02 & 504.80 & 516.00 & 524.30 & 505.22 & 516.00 & 523.53 & 508.62 & 515.99 & 525.41 \\
Chile           & 445.56 & 484.83 & 480.03 & 441.52 & 484.83 & 483.88 & 444.33 & 484.83 & 480.91 & 454.45 & 484.83 & 483.06 \\
Colombia        & 406.98 & 434.60 & 437.21 & 408.69 & 430.02 & 441.09 & 403.28 & 434.60 & 435.36 & 408.27 & 434.60 & 437.88 \\
\end{tabular}
\caption{Targets in performance scores on the boundary of the PPS ($\partial^{\text{W}}P$) considering instruction time, expenditure and ESCS as inputs, computed by radial and directional models using ACES under VRS. Threshold $\delta$ is taken as $1\%$ for direction (18) and $5\%$ for direction (20).}
\label{tab:targetsvrs134_ACES}
\end{sidewaystable}

\begin{sidewaystable}[ht]
\centering
\footnotesize
{
\setlength{\tabcolsep}{3pt}
\begin{tabular}{l|rrrrrrr|rrrrrrrr}
\multicolumn{1}{c|}{}    & \multicolumn{7}{c|}{No ESCS}                                                                                                                                                                              & \multicolumn{8}{c}{ESCS}                                                                                                                                                                                                                \\
\multicolumn{1}{c|}{DMU} & \multicolumn{1}{c}{Ireland} & \multicolumn{1}{c}{Japan} & \multicolumn{1}{c}{Korea} & \multicolumn{1}{c}{Latvia} & \multicolumn{1}{c}{Mexico} & \multicolumn{1}{c}{Poland} & \multicolumn{1}{c|}{Türkiye} & \multicolumn{1}{c}{Colombia} & \multicolumn{1}{c}{Ireland} & \multicolumn{1}{c}{Japan} & \multicolumn{1}{c}{Korea} & \multicolumn{1}{c}{Latvia} & \multicolumn{1}{c}{Mexico} & \multicolumn{1}{c}{Poland} & \multicolumn{1}{c}{Türkiye} \\ \hline
Australia                & 0.501                       & 0.499                     & 0.000                     & 0.000                      & 0.000                      & 0.000                      & 0.000                        & 0.000                        & 0.501                       & 0.499                     & 0.000                     & 0.000                      & 0.000                      & 0.000                      & 0.000                       \\
Austria                  & 0.000                       & 0.095                     & 0.905                     & 0.000                      & 0.000                      & 0.000                      & 0.000                        & 0.000                        & 0.000                       & 0.378                     & 0.416                     & 0.000                      & 0.000                      & 0.206                      & 0.000                       \\
\end{tabular}
}
\caption{Optimal values of $\lambda _j$ for each DMU, computed by a radial model with deterministic variables under VRS without considering and considering ESCS as an input. In the last row are the means of optimal $\lambda$ values for countries in reference sets of other countries, expressed as a percentage. This percentage quantifies the influence (or the relative importance) of the corresponding efficient DMU as a reference for the inefficient DMUs.}
\label{tab:lambdas}
\end{sidewaystable}

\begin{sidewaystable}[ht]
\centering
\footnotesize
\begin{tabular}{l|rrrrr|rrrrr}
\multicolumn{1}{c|}{}          & \multicolumn{5}{c|}{No ESCS}                                                                                                                                               & \multicolumn{5}{c}{ESCS}                                                                                                                                                  \\
\multicolumn{1}{c|}{\textbf{}} & \multicolumn{1}{c}{$\alpha = 0.3$} & \multicolumn{1}{c}{$\alpha = 0.35$} & \multicolumn{1}{c}{$\alpha = 0.4$} & \multicolumn{1}{c}{$\alpha = 0.45$} & \multicolumn{1}{c|}{$\alpha = 0.5$} & \multicolumn{1}{c}{$\alpha = 0.3$} & \multicolumn{1}{c}{$\alpha = 0.35$} & \multicolumn{1}{c}{$\alpha = 0.4$} & \multicolumn{1}{c}{$\alpha = 0.45$} & \multicolumn{1}{c}{$\alpha = 0.5$} \\ \hline
Australia                      & \textbf{1.000}                  & \textbf{1.000}                   & \textbf{1.000}                  & 0.996                            & 0.965                            & \textbf{1.000}                  & \textbf{1.000}                   & \textbf{1.000}                  & 0.996                            & 0.965                           \\
Austria                        & \textbf{1.000}                  & \textbf{1.000}                   & \textbf{1.000}                  & 0.968                            & 0.932                            & \textbf{1.000}                  & \textbf{1.000}                   & \textbf{1.000}                  & 0.981                            & 0.942                           \\
Belgium                        & \textbf{1.000}                  & \textbf{1.000}                   & 0.992                           & 0.962                            & 0.929                            & \textbf{1.000}                  & \textbf{1.000}                   & \textbf{1.000}                  & 0.972                            & 0.933                           \\
\end{tabular}
\caption{Chance-constrained efficiency scores $\rho ^*$ of output-oriented radial models under VRS for different values of $\alpha$, without considering and considering ESCS as an input. The stochastic variables are the performance scores and the ESCS. The case $\alpha = 0.5$ coincides with the deterministic case.}
\label{tab:cceff}
\end{sidewaystable}

\begin{table}[ht]
\centering
\footnotesize
\begin{tabular}{l|crrrrrrrrr}
\multicolumn{1}{c|}{DMU} &      & \multicolumn{1}{c}{$\alpha=0.2$} & \multicolumn{1}{c}{$0.3$} & \multicolumn{1}{c}{$0.4$} & \multicolumn{1}{c}{$0.5$} & \multicolumn{1}{c}{$0.6$} & \multicolumn{1}{c}{$0.7$} & \multicolumn{1}{c}{$0.8$} & \multicolumn{1}{c}{$0.9$} & \multicolumn{1}{c}{$1$} \\ \hline
Australia                 & L         & 0.551                         & 0.603                         & 0.660                         & 0.721                         & 0.765                         & 0.810                         & 0.858                         & 0.909                         & 0.962                       \\
                          & U         & \textbf{1.000}                & \textbf{1.000}                & \textbf{1.000}                & \textbf{1.000}                & \textbf{1.000}                & \textbf{1.000}                & \textbf{1.000}                & \textbf{1.000}                & 0.962                       \\
Austria                   & L         & 0.549                         & 0.596                         & 0.647                         & 0.700                         & 0.740                         & 0.781                         & 0.827                         & 0.877                         & 0.928                       \\
                          & U         & \textbf{1.000}                & \textbf{1.000}                & \textbf{1.000}                & \textbf{1.000}                & \textbf{1.000}                & \textbf{1.000}                & \textbf{1.000}                & 0.983                         & 0.928                       \\
Belgium                   & L         & 0.547                         & 0.595                         & 0.646                         & 0.700                         & 0.740                         & 0.782                         & 0.827                         & 0.876                         & 0.927                       \\
\end{tabular}
\caption{Lower (L) (worst) and upper (U) (best) efficiency scores $\rho ^*$ for different values of $\alpha$ applying Kao--Liu meta-model with radial model under VRS, considering instruction time and expenditure as inputs.}
\label{tab:fuzzy_12}
\end{table}

\begin{table}[ht]
\centering
\footnotesize
\begin{tabular}{l|crrrrrrrrr}
\multicolumn{1}{c|}{DMU} &      & \multicolumn{1}{c}{$\alpha=0.2$} & \multicolumn{1}{c}{$0.3$} & \multicolumn{1}{c}{$0.4$} & \multicolumn{1}{c}{$0.5$} & \multicolumn{1}{c}{$0.6$} & \multicolumn{1}{c}{$0.7$} & \multicolumn{1}{c}{$0.8$} & \multicolumn{1}{c}{$0.9$} & \multicolumn{1}{c}{$1$} \\ \hline
Australia                 & L         & 0.551                         & 0.603                         & 0.660                         & 0.721                         & 0.765                         & 0.810                         & 0.858                         & 0.909                         & 0.962                       \\
                          & U         & \textbf{1.000}                & \textbf{1.000}                & \textbf{1.000}                & \textbf{1.000}                & \textbf{1.000}                & \textbf{1.000}                & \textbf{1.000}                & \textbf{1.000}                & 0.962                       \\
Austria                   & L         & 0.556                         & 0.604                         & 0.655                         & 0.709                         & 0.750                         & 0.792                         & 0.836                         & 0.887                         & 0.939                       \\
                          & U         & \textbf{1.000}                & \textbf{1.000}                & \textbf{1.000}                & \textbf{1.000}                & \textbf{1.000}                & \textbf{1.000}                & \textbf{1.000}                & 0.994                         & 0.939                       \\
Belgium                   & L         & 0.550                         & 0.598                         & 0.649                         & 0.703                         & 0.744                         & 0.786                         & 0.831                         & 0.880                         & 0.932                       \\
\end{tabular}
\caption{Lower (L) (worst) and upper (U) (best) efficiency scores $\rho ^*$ for different values of $\alpha$ applying Kao--Liu meta-model with radial model under VRS, considering instruction time, expenditure and ESCS as inputs.}
\label{tab:fuzzy_123}
\end{table}

\clearpage

\section*{Kao-Liu fuzzy meta-model figures}

\begin{figure}[ht]
\centering
  \includegraphics[width=0.91\textwidth]{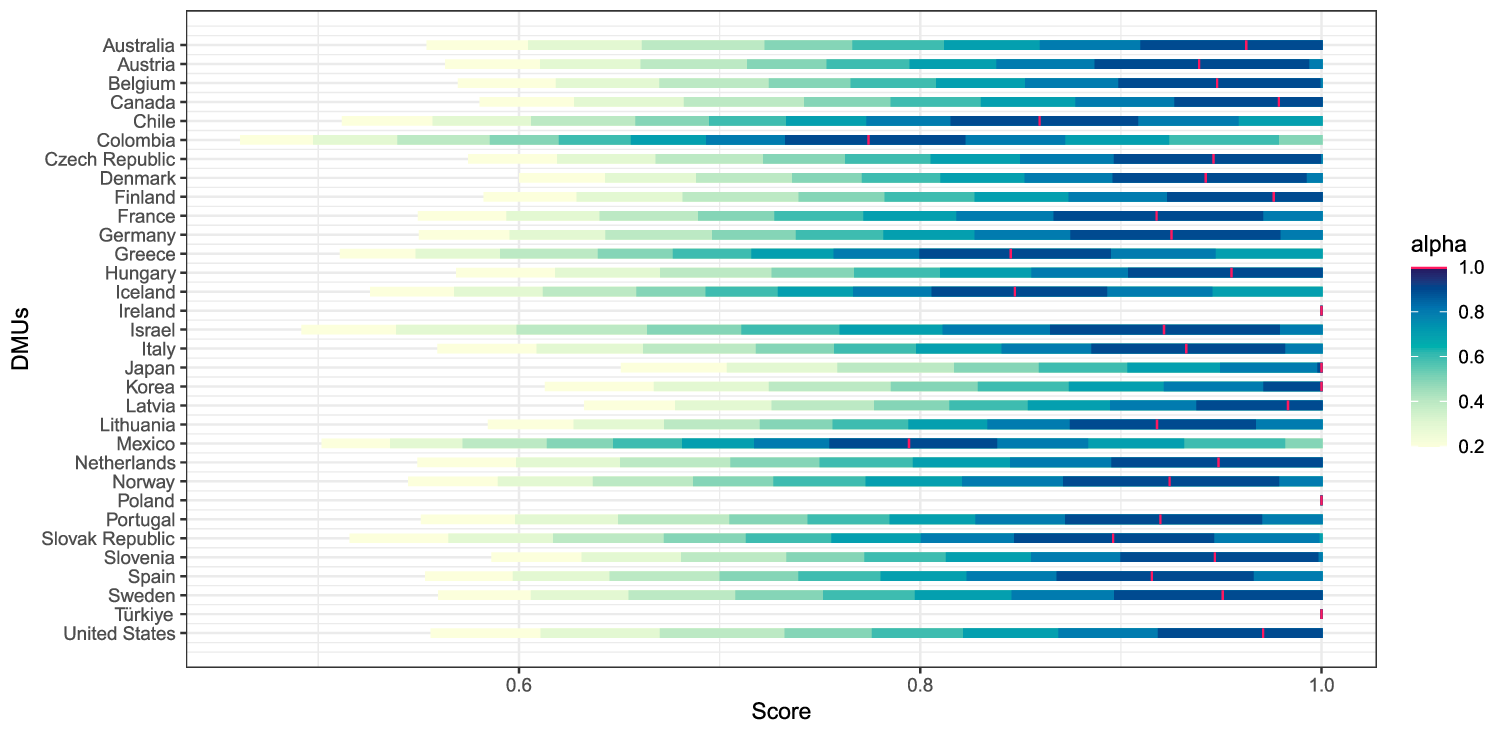}
\caption{Efficiency scores $\rho ^*$ applying Kao--Liu fuzzy meta-model with radial model under VRS, considering instruction time and expenditure as inputs. It shows the same information as Table~\ref{tab:fuzzy_12}.}
\label{fig:fuzzy12}
\end{figure}

\begin{figure}[hb]
\centering
  \includegraphics[width=0.91\textwidth]{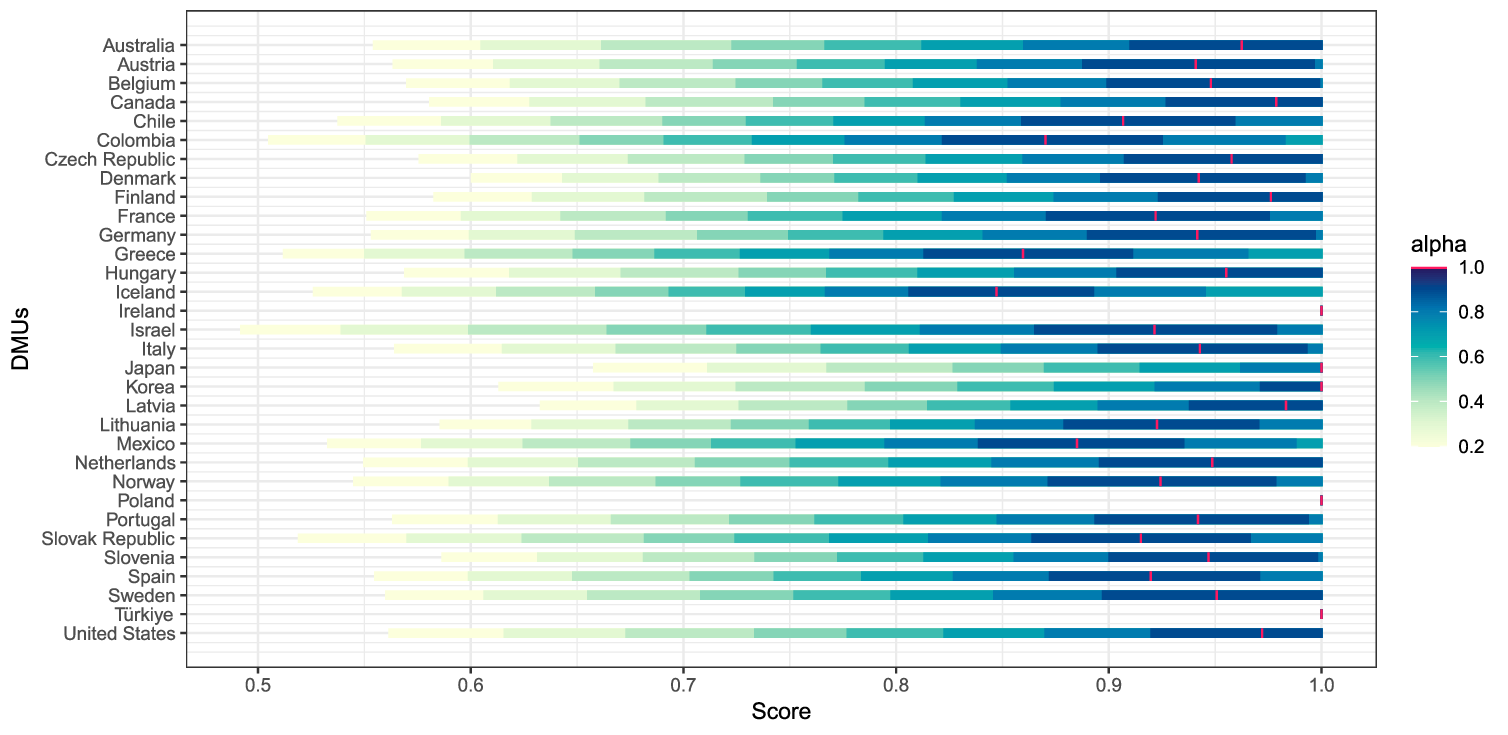}
\caption{Efficiency scores $\rho ^*$ applying Kao--Liu fuzzy meta-model with radial model under VRS, considering instruction time, expenditure, and ESCS as inputs. It shows the same information as Table~\ref{tab:fuzzy_123}.}
\label{fig:fuzzy123}
\end{figure}

\clearpage

\bibliography{refs}

@article{alcacer2025survey,
  title={A survey on archetypal analysis},
  author={Alcacer, Aleix and Epifanio, Irene and Mair, Sebastian and M{\o}rup, Morten},
  journal={arXiv preprint arXiv:2504.12392},
  year={2025},
  doi = {10.48550/arXiv.2504.12392}
}

@article{Ali90,
title = {Translation invariance in data envelopment analysis},
journal = {Operations Research Letters},
volume = {9},
number = {6},
pages = {403-405},
year = {1990},
doi = {10.1016/0167-6377(90)90061-9},
author = {Agha {Iqbal Ali} and Lawrence M. Seiford}
}

@article{Andersen1993,
author = {Andersen, Per and Petersen, Niels Christian},
doi = {10.1287/mnsc.39.10.1261},
journal = {Management Science},
number = {10},
pages = {1261--1264},
title = {A Procedure for Ranking Efficient Units in Data Envelopment Analysis},
volume = {39},
year = {1993}
}

@article{Banker1984,
author = {Banker, R. D. and Charnes, A. and Cooper, W. W.},
title = {Some Models for Estimating Technical and Scale Inefficiencies in Data Envelopment Analysis},
journal = {Management Science},
volume = {30},
number = {9},
pages = {1078-1092},
year = {1984},
doi = {10.1287/mnsc.30.9.1078}
}

@article{Bas2026,
title = {{ENCERTIA}: A dynamic {R}-shiny app to support business decision-making using data envelopment analysis},
journal = {SoftwareX},
volume = {33},
pages = {102525},
year = {2026},
doi = {10.1016/j.softx.2026.102525},
author = {Mar\'{\i}a C. Bas and Rafael Ben\'{\i}tez and Vicente J. Bol\'os},
}

@article{BM86,
 author = {Rajiv D. Banker and Richard C. Morey},
 journal = {Operations Research},
 number = {4},
 pages = {513-521},
 title = {Efficiency Analysis for Exogenously Fixed Inputs and Outputs},
 volume = {34},
 year = {1986}
}

@article{Bolos2024,
author = {Bol\'os, V. J. and Ben\'{\i}tez, R. and Coll-Serrano, V.},
title = {Chance constrained directional models in stochastic data envelopment analysis},
journal = {Operations Research Perspectives},
volume = {12},
year = {2024},
pages = {100307},
doi = {10.1016/j.orp.2024.100307}
}

@article{Bolos2025,
author = {Bol\'os, V. J. and Ben\'{\i}tez, R. and Coll-Serrano, V.},
title = {A new family of models with generalized orientation in data envelopment analysis},
journal = {International Transactions in Operational Research},
volume = {early access},
year = {2026},
doi = {10.1111/itor.70063}
}

@article{Briec1997,
author = {Briec, W.},
title = {A Graph-Type Extension of {Farrell} Technical Efficiency Measure},
journal = {Journal of Productivity Analysis},
volume = {8},
number = {1},
pages = {95-110},
year  = {1997},
doi={10.1023/A:1007728515733}
}

@article{Buy2022,
title = {Analyzing {TALIS} Indicators and {PISA} Results with Data Envelopment: Comparison of {EMS}, {DEAP} and {R} Software},
volume = {9},
doi = {10.33200/ijcer.1038281},
number = {3},
journal={International Journal of Contemporary Educational Research},
author={B\"uy\"ukkidik, Serap},
year={2022},
pages={492–508}
}

@article{Chambers1996,
title = {Benefit and Distance Functions},
journal = {Journal of Economic Theory},
volume = {70},
number = {2},
pages = {407-419},
year = {1996},
doi={10.1006/jeth.1996.0096},
author = {Chambers, R. G. and Chung, Y. and F\"are, R.}
}

@article{Chambers1998,
author = {Chambers, R. G. and Chung, Y. and F\"are, R.},
doi={10.1023/A:1022637501082},
journal = {Journal of Optimization Theory and Applications},
number = {2},
pages = {351-364},
title = {Profit, Directional Distance Functions, and {Nerlovian} Efficiency},
volume = {98},
year = {1998}
}

@article{Charnes1978,
author = {Charnes, A. and Cooper, W. W. and Rhodes, E.},
doi={10.1016/0377-2217(78)90138-8},
journal = {European Journal of Operational Research},
number = {6},
pages = {429-444},
title = {Measuring the Efficiency of Decision Making Units},
volume = {2},
year = {1978}
}

@article{Charnes1979,
title = {Short Communication: Measuring the Efficiency of Decision Making Units},
journal = {European Journal of Operational Research},
volume = {3},
number = {4},
pages = {339},
year = {1979},
doi = {10.1016/0377-2217(79)90229-7},
author = {A. Charnes and W.W. Cooper and E. Rhodes}
}

@article{Charnes1981,
author = {Charnes, A. and Cooper, W. W. and Rhodes, E.},
doi = {10.1287/mnsc.27.6.668},
journal = {Management Science},
number = {6},
pages = {668--697},
title = {Evaluating Program and Managerial Efficiency: an Application of Data Envelopment Analysis to {Program Follow Through}},
volume = {27},
year = {1981}
}

@article{Charnes1985,
author = {Charnes, A and Cooper, W.W and Golany, B and Seiford, L and Stutz, J},
doi = {10.1016/0304-4076(85)90133-2},
journal = {Journal of Econometrics},
number = {1-2},
pages = {91--107},
title = {Foundations of Data Envelopment Analysis for {Pareto-Koopmans} Efficient Empirical Production Functions},
volume = {30},
year = {1985}
}

@article{Cooper1996,
author = {W W Cooper and Z M Huang and S X Li},
title = {Satisficing {DEA} Models under Chance Constraints},
journal = {Annals of Operations Research},
volume = {66},
pages = {279-295},
year  = {1996},
doi = {10.1007/BF02187302}
}

@article{Cooper1998,
author = {W W Cooper and Z Huang and V Lelas and S X Li and O B Olesen},
title = {Chance Constrained Programming Formulations for Stochastic Characterizations of Efficiency and Dominance in {DEA}},
journal = {Journal of Productivity Analysis},
volume = {9},
number = {1},
pages = {53-79},
year  = {1998},
doi = {10.1023/A:1018320430249}
}

@article{Cooper2002,
author = {W W Cooper and H Deng and Z Huang and S X Li},
title = {Chance constrained programming approaches to technical efficiencies and inefficiencies in stochastic data envelopment analysis},
journal = {Journal of the Operational Research Society},
volume = {53},
number = {12},
pages = {1347-1356},
year  = {2002},
doi = {10.1057/palgrave.jors.2601433}
}

@book{Cooper2007,
author = {Cooper, William W and Seiford, Lawrence M and Tone, Kaoru},
doi = {10.1007/978-0-387-45283-8},
isbn = {9780387452814},
issn = {0160-5682},
pages = {483},
publisher = {Springer},
edition = {2nd},
title = {Data Envelopment Analysis. A Comprehensive Text with Models, Applications, References and {DEA}-Solver Software},
year = {2007}
}

@article{Cilin2018,
title = {Testing the Validity and Structure of the Data Envelopment Analysis {PISA} Scores},
journal={International Journal of Advances in Management and Economics}, author={\c{C}ilingirt\"urk, A Mete and Turanli, M\"unevver},
year={2018}
}

@article{deaRSX,
title = {{deaR}: Conventional and fuzzy {DEA} models with {R}},
journal = {SoftwareX},
volume = {31},
pages = {102266},
year = {2025},
doi = {10.1016/j.softx.2025.102266},
author = {Bol\'os, V. J. and Ben\'{\i}tez, R. and Coll-Serrano, V.}
}

@Manual{deaR25,
title = {{deaR}: Conventional and Fuzzy Data Envelopment Analysis},
author = {Coll-Serrano, V. and Bol\'os, V. J. and Ben\'{\i}tez, R.},
year = {2025},
note = {R package version 1.5.2},
doi = {10.32614/CRAN.package.deaR},
url = {https://CRAN.R-project.org/package=deaR}
}

@article{Deb51,
author = {Debreu, G.},
title = {The coefficient of resource utilization},
journal = {Econometrica},
volume = {9},
pages = {273-292},
year = {1951},
doi = {10.2307/1906814}
}

@article{Emr09,
author = {Emrouznejad, Ali and Amin, Gholam R.},
year = {2009},
month = {01},
pages = {486-498},
title = {DEA models for ratio data: Convexity consideration},
volume = {33},
journal = {Applied Mathematical Modelling},
doi = {10.1016/j.apm.2007.11.018}
}

@article{Far57,
author = {Farrell, M. J.},
title = {The measurement of productive efficiency},
journal = {Journal of the Royal Statistical Society, Series A (General)},
volume = {120},
number = {3},
pages = {253-290},
year = {1957},
doi = {10.2307/2343100}
}

@article{Fare1978,
author = {Rolf F{\"a}re and C. A. {Knox Lovell}},
title = {Measuring the Technical Efficiency of Production},
journal = {Journal of Economic Theory},
volume = {19},
number = {1},
pages = {150-162},
year = {1978},
doi = {10.1016/0022-0531(78)90060-1}
}

@book{Fare1985,
author = {F\"are, R. and Grosskopf, S. and Lovell, C. A. K.},
title = {The Measurement of Efficiency of Production},
series = {Studies in Productivity Analysis},
year = {2013},
publisher = {Springer Dordrecht},
doi = {10.1007/978-94-015-7721-2}
}

@article{Golany1997,
journal={Socio-Economic Planning Sciences},
author={Golany, Boaz and Thore, Sten},
title={The Economic and Social Performance of Nations: Efficiency and Returns to Scale},
year={1997},
pages={191-204},
volume={31},
number={3}
}

@article{Guo2001,
title = {Fuzzy {DEA}: a perceptual evaluation method},
journal = {Fuzzy Sets and Systems},
volume = {119},
number = {1},
pages = {149-160},
year = {2001},
doi = {10.1016/S0165-0114(99)00106-2},
author = {Peijun Guo and Hideo Tanaka}
}

@article{Hatami19,
title = {Data envelopment analysis models with ratio data: A revisit},
author = {Adel Hatami-Marbini and Mehdi Toloo},
year = {2019},
doi = {10.1016/j.cie.2019.04.041},
volume = {133},
pages = {331-338},
journal = {Computers and Industrial Engineering}
}

@article{Hollin03,
author = {B. Hollingsworth and P. Smith},
title = {Use of ratios in data envelopment analysis},
journal = {Applied Economics Letters},
volume = {10},
number = {11},
pages = {733-735},
year = {2003},
doi = {10.1080/1350485032000133381}
}

@article{Kao2000,
title = {Fuzzy efficiency measures in data envelopment analysis},
journal = {Fuzzy Sets and Systems},
volume = {113},
number = {3},
pages = {427-437},
year = {2000},
doi = {10.1016/S0165-0114(98)00137-7},
author = {Chiang Kao and Shiang-Tai Liu}
}

@inproceedings{Kocak2011,
author = {Kocak, Habip and \c{C}ilingirt\"urk, Ahmet},
year = {2011},
month = {04},
booktitle = {Annual International Conference on Operations Research and Statistics},
pages = {83--88},
title = {Efficiency Analysis of {OECD} Public Education Spending}
}

@incollection{Koop51,
author = {Koopmans, T. C.},
title = {Analysis of production as an efficient combination of activities},
editor = {Koopmans, T. C.},
publisher = {Cowles Commission for Research in Economics},
booktitle = {Activity Analysis of Production and Allocation},
series = {13},
pages = {33-97},
year = {1951}
}

@article{Lan1993,
author = {Land, K. C. and Knox Lovell, C. A. and Thore, S.},
title = {Chance-constrained data envelopment analysis},
journal = {Managerial and Decision Economics},
year = {1993},
volume = {14},
number = {6},
pages = {541--554},
note = {https://www.jstor.org/stable/2487873}
}

@article{Leon2003,
title = {A fuzzy mathematical programming approach to the assessment of efficiency with {DEA} models},
journal = {Fuzzy Sets and Systems},
volume = {139},
number = {2},
pages = {407-419},
year = {2003},
doi = {10.1016/S0165-0114(02)00608-5},
author = {T. León and V. Liern and J.L. Ruiz and I. Sirvent}
}

@article{Lovell94,
author = {Lovell, C.A.K.},
year = {1994},
pages = {75-248},
title = {" Linear Programming Approaches to the Measurement and Analysis of Productive Efficiency"},
volume = {2},
journal = {{TOP}}
}

@article{Olesen1995,
author = {Olesen, O. B. and Petersen, N. C.},
title = {Chance Constrained Efficiency Evaluation},
journal = {Management Science},
volume = {41},
number = {3},
pages = {442-457},
year = {1995},
doi = {10.1287/mnsc.41.3.442}
}

@article{Olesen2015,
title = {Efficiency analysis with ratio measures},
journal = {European Journal of Operational Research},
volume = {245},
number = {2},
pages = {446-462},
year = {2015},
doi = {10.1016/j.ejor.2015.03.013},
author = {Ole Bent Olesen and Niels Christian Petersen and Victor V. Podinovski}
}

@book{pisa2022,
  title     = {{PISA} 2022 Results (Volume I): The State of Learning and Equity in Education},
  author    = {{OECD}},
  year      = {2023},
  publisher = {OECD Publishing},
  address   = {Paris},
  series    = {PISA},
  doi       = {10.1787/53f23881-en},
  url       = {https://doi.org/10.1787/53f23881-en}
}

@book{glance2023,
  title     = {Education at a Glance 2023: {OECD} Indicators},
  author    = {{OECD}},
  year      = {2023},
  publisher = {OECD Publishing},
  address   = {Paris},
  doi       = {10.1787/e13bef63-en},
  url       = {https://doi.org/10.1787/e13bef63-en}
}

@book{glance2025,
  title     = {Education at a Glance 2025: {OECD} Indicators},
  author    = {{OECD}},
  year      = {2025},
  publisher = {OECD Publishing},
  address   = {Paris},
  doi       = {10.1787/1c0d9c79-en},
  url       = {https://doi.org/10.1787/1c0d9c79-en}
}

@Manual{R25,
title = {R: A Language and Environment for Statistical Computing},
author = {{R Core Team}},
organization = {R Foundation for Statistical Computing},
address = {Vienna, Austria},
year = {2025},
url = {https://www.R-project.org/}
}

@article{Ruggiero96,
title = {On the measurement of technical efficiency in the public sector},
journal = {European Journal of Operational Research},
volume = {90},
number = {3},
pages = {553-565},
year = {1996},
doi = {10.1016/0377-2217(94)00346-7},
author = {John Ruggiero}
}

@Manual{SdeaR25,
title = {{SdeaR}: Stochastic Data Envelopment Analysis},
author = {Bol\'os, V. J. and Coll-Serrano, V. and Ben\'{\i}tez, R.},
year = {2025},
note = {R package version 1.0.1},
doi = {10.32614/CRAN.package.SdeaR},
url = {https://CRAN.R-project.org/package=SdeaR}
}

@article{Sirin2005,
title = {Socioeconomic status and academic achievement: A meta-analytic review of research},
journal = {Review of Educational Research},
volume = {75},
pages = {417-453},
year = {2005},
doi = {10.3102/00346543075003417},
author = {Sirin, S. R.}
}

@article{Tavana2013,
title = {Chance-constrained {DEA} models with random fuzzy inputs and outputs},
journal = {Knowledge-Based Systems},
volume = {52},
pages = {32-52},
year = {2013},
doi = {10.1016/j.knosys.2013.05.014},
author = {Madjid Tavana and Rashed Khanjani Shiraz and Adel Hatami-Marbini and Per J. Agrell and Khalil Paryab}
}

@article{Tone2001,
author = {Tone, Kaoru},
doi={10.1016/S0377-2217(99)00407-5},
journal = {European Journal of Operational Research},
number = {3},
pages = {498--509},
title = {A Slacks-Based Measure of Efficiency in Data Envelopment Analysis},
volume = {130},
year = {2001}
}

@article{Tone2016,
author = {Kaoru Tone},
title = {Data Envelopment Analysis as a Kaizen Tool: SBM Variations Revisited},
year = {2016},
month = {8},
volume = {16},
pages = {49--61},
journal = {Bulletin of Mathematical Sciences and Applications},
doi = {10.18052/www.scipress.com/BMSA.16.49}
}

@article{espana2024,
  title={Estimating production functions through additive models based on regression splines},
  author={Espa{\~n}a, Victor J and Aparicio, Juan and Barber, Xavier and Esteve, Miriam},
  journal={European Journal of Operational Research},
  volume={312},
  number={2},
  pages={684--699},
  year={2024},
  publisher={Elsevier}
}

@article{espana2025,
  title={Estimating production technologies using multi-output adaptive constrained enveloping splines},
  author={Espa{\~n}a, Victor J and Aparicio, Juan and Barber, Xavier},
  journal={Computers \& Operations Research},
  pages={107242},
  year={2025},
  publisher={Elsevier}
}

@article{espana2025rf,
  title={An adaptation of Random Forest to estimate convex non-parametric production technologies: an empirical illustration of efficiency measurement in education},
  author={Espa{\~n}a, Victor J and Aparicio, Juan and Barber, Xavier},
  journal={International Transactions in Operational Research},
  volume={32},
  number={5},
  pages={2523--2546},
  year={2025},
  doi = {10.1111/itor.13561},
  publisher={Wiley Online Library}
}

@article{friedman1991,
  title={Multivariate adaptive regression splines},
  author={Friedman, Jerome H},
  journal={The annals of statistics},
  volume={19},
  number={1},
  pages={1--67},
  year={1991},
  publisher={Institute of Mathematical Statistics}
}

@article{esteve2020,
  title={Efficiency analysis trees: A new methodology for estimating production frontiers through decision trees},
  author={Esteve, Miriam and Aparicio, Juan and Rabasa, Alejandro and Rodriguez-Sala, Jesus J},
  journal={Expert Systems with Applications},
  volume={162},
  pages={113783},
  year={2020},
  publisher={Elsevier}
}

@article{guillen2023,
  title={Gradient tree boosting and the estimation of production frontiers},
  author={Guillen, Maria D and Aparicio, Juan and Esteve, Miriam},
  journal={Expert systems with applications},
  volume={214},
  pages={119134},
  year={2023},
  publisher={Elsevier}
}

@article{valero2021,
  title={Support vector frontiers: A new approach for estimating production functions through support vector machines},
  author={Valero-Carreras, Daniel and Aparicio, Juan and Guerrero, Nadia M},
  journal={Omega},
  volume={104},
  pages={102490},
  year={2021},
  publisher={Elsevier}
}

@article{moragues2023,
  title={Measuring technical efficiency for multi-input multi-output production processes through OneClass Support Vector Machines: a finite-sample study},
  author={Moragues, Raul and Aparicio, Juan and Esteve, Miriam},
  journal={Operational Research},
  volume={23},
  number={3},
  pages={47},
  year={2023},
  publisher={Springer}
}
\bibliographystyle{abbrvnat}

\end{document}